\documentclass[12pt]{article}
\usepackage{amsmath, amsthm, amsfonts, amssymb}
\usepackage{graphicx,psfrag,epsf}
\usepackage{enumerate}
\usepackage{natbib}
\usepackage{bbm}
\usepackage{hyperref}
\usepackage{cleveref}
\usepackage{bm}
\usepackage{algorithm}
\usepackage{algpseudocode}
\usepackage{booktabs}
\usepackage{multicol}
\usepackage{multirow}
\usepackage{xcolor}
\definecolor{codexdarkgreen}{RGB}{0,100,0}
\colorlet{green}{codexdarkgreen}
\usepackage{mathtools}
\usepackage{url} 

\newcommand{\blind}{1}

\theoremstyle{plain}
\newtheorem{theorem}{Theorem}[]
\newtheorem{corollary}{Corollary}[theorem]
\newtheorem{assumption}{Assumption}[]
\newtheorem{proposition}{Proposition}%
\theoremstyle{definition}

\theoremstyle{remark}

\def\R{\mathbb{R}}

\def\bX{\bm{X}}

\def\bV{\bm{V}}
\def\bU{\bm{U}}
\def\bZ{\bm{Z}}

\def\tr{\textnormal{tr}}

\renewcommand{\P}{\mathrm{P}}
\newcommand{\E}{\mathbb{E}}

\newcommand{\ind}{\perp\!\!\!\perp}

\newcommand{\htmu}{\hat{\tilde\mu}}
\newcommand{\httau}{\hat{\tilde\tau}}

\newcommand{\cM}{\mathcal{M}}
\newcommand{\cD}{\mathcal{D}}
\newcommand{\cG}{\mathcal{G}}
\newcommand{\cH}{\mathcal{H}}

\newcommand{\Rad}{\mathfrak{R}}

\begin{document}

\def\spacingset#1{\renewcommand{\baselinestretch}%
{#1}\small\normalsize} \spacingset{1}


\if1\blind
{
  \title{\bf Improving RCT-Based CATE Estimation Under Covariate Mismatch via Double Calibration}
  \author{Samhita Pal,\thanks{
    samhita.pal@vumc.org}\hspace{.2cm}\\
    Department of Biostatistics, Vanderbilt University Medical Center\\
    Jared D. Huling \\
    Division of Biostatistics and Health Data Science, University of Minnesota\\
    and\\
    Amir Asiaee\\
    Department of Biostatistics, Vanderbilt University Medical Center\\}
  \maketitle
} \fi

\if0\blind
{
  \bigskip
  \bigskip
  \bigskip
  \begin{center}
    {\Large\bf Improving RCT-Based CATE Estimation Under Covariate Mismatch via Double Calibration}
\end{center}
  \medskip
} \fi

\bigskip
\begin{abstract}
We develop estimators that improve precision of heterogeneous treatment effect estimates that allow borrowing information from observational studies when the available covariates in each data source do not perfectly match. Standard data-borrowing methods often assume perfectly matched covariates. We propose MR-OSCAR, an RCT-calibrated, two-stage estimation approach that first predicts the trial-missing variables using the observational data via imputation and then calibrates observational outcome predictions to the randomized trial, preserving the causal contrast, unlike the results for generalization, where imputation does not improve performance. Our theory gives finite-sample guarantees with a transparent error decomposition including an imputation error that shrinks as the observational mapping becomes more predictable. Simulations show that imputation almost always outperforms naively using only the shared covariates and clarifies when borrowing helps (strong predictability of the missing block, moderate trial size) and when it does not (poor predictability or dominant trial-only moderators). We motivate the approach with the Greenlight Plus trial on early childhood obesity and outline a forthcoming EHR analysis at Vanderbilt, highlighting the use of our method in common scenarios where data do not perfectly align.
\end{abstract}

\noindent%
\vfill

\newpage
\spacingset{1.45} 

{
\section{Introduction}
\label{sec:intro}

Heterogeneous treatment effects (HTEs) are central to precision prevention and care. 
Randomized controlled trials (RCTs) remain the gold standard for causal inference, 
yet most trials are powered for average effects rather than the fine-grained 
heterogeneity that guides individualized decisions 
\citep{wang2007statistics, kent2018personalized}. 
Large observational studies (OS) contain rich covariates and massive sample sizes, 
but are vulnerable to confounding and design biases. 
This asymmetry motivates borrowing information from observational data to 
improve {within-trial} precision for conditional average treatment effects (CATEs), 
provided the borrowing is conducted such that observational confounding cannot distort 
the randomized contrast. A growing methodological literature has made this idea concrete by showing how to combine trials with large observational cohorts while preserving, or even sharpening, the randomized contrast \cite{cheng2021adaptive, oberst2022understanding, raman2023optimizing}. Recent causal fusion frameworks use weighting, outcome modeling, or double-robust estimators to integrate an RCT with an external cohort under formal transportability assumptions, and theory and simulations show that, when those assumptions hold, borrowing can substantially reduce variance for HTE estimands relative to trial-only analyses \cite{asiaee2023leveraging, wu2022integrative}.

In this paper, we aim to improve the precision of {RCT-based} CATE 
estimation by borrowing outcome-prediction structure from a large observational 
database while ensuring that the causal identification remains anchored in the 
randomized contrast. Several existing approaches follow this principle. 
Representation-learning methods 
\citep{johansson2016learning, yao2018representation, hatt2022combining} use large OS 
datasets to learn prognostic scores or balanced embeddings that stabilize effect 
estimation in small RCTs. Most directly related, \cite{asiaee2023leveraging} propose 
R-OSCAR, a two-stage estimator that learns outcome models from OS data, calibrates 
them to the RCT via a discrepancy function estimated on the trial, and uses the 
calibrated predictions as variance-reducing pseudo-outcomes. Because the final CATE 
is identified solely through the randomized contrast, OS confounding does not bias the 
estimate. In parallel, \cite{karlsson2025robust} introduce the QR-learner, a 
model-agnostic learner that leverages external data for CATE estimation in the trial 
population while guaranteeing that using external data cannot worsen the risk 
relative to a trial-only learner. For a broader overview of methods that integrate 
trials and non-experimental data for HTE, see \cite{brantner2023methods}.

A major barrier to such integration is that the covariates collected in trials and 
observational sources rarely align. We refer to this as 
{covariate mismatch}: the RCT may contain behavioral or survey-based variables 
absent from the OS, while the OS may contain laboratory or utilization features not 
recorded in the trial \citep{Rassler2002}. This differs fundamentally from 
{covariate shift} \citep{sugiyama2012machine, ghosh2026demographic}, where the {same} covariates 
are observed in both sources but follow different distributions. Under mismatch, 
standard fusion strategies including regression adjustment, reweighting, 
and many recent trial-OS generalizability estimators 
\citep{ackerman2019implementing, josey2022calibration, li2022generalizing, colnet2021generalizing} can fail 
because key effect modifiers or prognostic variables are unobserved in one of the 
sources. Indeed, \cite{colnet2021generalizing} show that when transporting trial 
effects to a target population, imputing covariates that were never measured in the 
trial cannot recover the required identification, even under perfectly specified 
linear models.



Under covariate mismatch, a few transfer-learning approaches address partial overlap in covariate sets \citep{bica2022transfer, chang2024heterogeneous, xu2025representation}, 
but typically rely on learned latent representations optimized for predictive fit. 
While only \cite{bica2022transfer} focuses on a causal setting, in general, such end-to-end procedures may not preserve the randomization structure and therefore 
provide limited guarantees that observational confounding will not leak into the 
CATE, a phenomenon known as \emph{negative transfer}. More broadly, power-likelihood and borrowing frameworks
\citep{lin2025powerlikelihood, rahman2021leveraging} combine
experimental and observational data for treatment effect estimation,
but do not explicitly address covariate mismatch or provide CMO-based
negative-transfer guarantees.

This paper extends R-OSCAR to settings with covariate mismatch. We partition baseline covariates into RCT-only, OS-only, and shared blocks and introduce discrepancy functions calibrated on the randomized data, explicitly accommodating outcome shift without assuming transportability of outcome means or CATEs \citep{dahabreh2020extending}. We consider two borrowing strategies: SR-OSCAR, which applies R-OSCAR to the shared covariates only, and our main proposal MR-OSCAR, which uses the OS to predict OS-only covariates from the shared block, imputes them in the RCT, fits OS outcome models on the full covariate set, calibrates to the RCT, and estimates CATEs via augmented pseudo-outcome regression. In the calibration step, both shared and RCT-only covariates can enter, so strong RCT-only predictors are retained.

On the theoretical side, we derive finite-sample risk bounds that decompose into an approximation term, a complexity term for the OS-trained nuisance functions, and an imputation term quantifying uncertainty from predicting OS-only covariates in the RCT. These bounds show that MR-OSCAR controls negative transfer: when imputation error is large, the risk is driven by the RCT calibration step and does not substantially deteriorate relative to RCT-only estimation. Specializing to sparse linear models yields concrete conditions characterizing when each strategy is beneficial. Simulations confirm these trade-offs across a range of covariate-mismatch configurations.

We illustrate these ideas using the Greenlight Plus early childhood obesity prevention trial \citep{heerman2022greenlight, heerman2024digital}, comparing standard counseling \citep{sanders2014greenlight} with an enhanced digital program. The trial was underpowered for fine-grained heterogeneity \citep{wang2007statistics, kent2018personalized}, but linked EHR data from non-enrolled children provide rich auxiliary information. Behavioral covariates in the RCT and utilization data in the EHR create a canonical covariate-mismatch setting in which MR-OSCAR sharpens CATE estimates while retaining protection against observational confounding.

The rest of the paper is organized as follows. \Cref{sec:meth} presents the setup, causal assumptions, and the proposed methods. \Cref{sec:theory} derives the error bounds and their sparse linear specialization. \Cref{sec:simu} reports simulation studies, and \Cref{sec:real_data} applies MR-OSCAR to the Greenlight Plus study. \Cref{sec:conclusion} concludes.



}

{
\section{Methods for data borrowing from an RCT under covariate mismatch}
\label{sec:meth}

In this section we formalize the covariate-mismatch setting, review the 
counterfactual mean outcome (CMO) and pseudo-outcome representation from 
\cite{asiaee2023leveraging}, and then develop two estimators that leverage 
observational information under misaligned covariates.

\subsection{Setup, background, and causal assumptions}
\label{subsec:notation}

We consider two data sources: a randomized trial and a large observational study. Let \(S\in\{r,o\}\) denote membership in the RCT and OS respectively, \(Y\) the outcome of interest, 
and treatment \(A\in\{-1,1\}\). For each unit we posit potential outcomes 
\(Y(1),Y(-1)\), the potential outcome under treatment and that under control, respectively.
Each source observes a different subset of baseline covariates: the RCT records \(\bX^r\) and the OS records \(\bX^o\). We denote the covariates measured in both sources as \(\bZ\in\mathbb R^{p_z}\) (the shared block), those exclusive to the RCT as \(\bU:=\bX^r\!\setminus\!\bZ\in\mathbb R^{p_u}\), and those exclusive to the OS by \(\bV:=\bX^o\!\setminus\!\bZ\in\mathbb R^{p_v}\), so that \(\bX^r=(\bU,\bZ)\in\mathbb R^{p_u+p_z}\) and \(\bX^o=(\bZ,\bV)\in\mathbb R^{p_z+p_v}\). We write \(\bX=(\bU,\bZ,\bV)\in\mathbb R^{p}\) with \(p=p_u+p_z+p_v\) for the complete covariate vector. For generic realizations we write \(\bm x=(\bm u,\bm z,\bm v)\), \(\bm x^r=(\bm u,\bm z)\), and \(\bm x^o=(\bm z,\bm v)\). Because covariate distributions can differ across sources, we retain the source indicator \(S=s\) in conditional distributions and expectations as needed. Since \(\bZ\) is a subvector of both \(\bX^r\) and \(\bX^o\), functions of \(\bZ\) may appear alongside functions of \(\bX^r\) or \(\bX^o\) in the same expression; in such cases \(\bZ\) denotes the shared component extracted from the source-specific vector. A complete glossary of notation is provided in Table~\ref{tab:notation} in the Supplement.
The RCT sample is \(\{(\bX_i^r,A_i,Y_i, S_i=r\}^{n^r}_{i=1}\), and the OS sample is \(\{(\bX_i^o,A_i,Y_i, S_i=o\}^{n^o}_{i=1}\).
Let \(n_a^s\) denote the number of units in arm \(a\) under source \(s\).  
For each study $s \in \{r,o\}$ and treatment arm $a \in \{-1,1\}$, we write 
$\mu^s_a(\bm x^s)$ $ \coloneqq \E[Y \mid A = a, \bX^s = \bm x^s, S = s]$ for the mean outcome that would be observed in population $s$ 
for individuals with observed covariates $\bm x^s$ if they all received treatment $a$, 
and $\pi^r_a(\bm x^r)$ for the randomization probability of arm $a$ in the RCT 
at covariate value $\bm x^r$. For source $S = s$, we write the per-arm \emph{full conditional mean outcome} as
\(
\mu^s_a(\bm x)
   = \E\bigl[Y \mid \bX=\bm x,A=a,S=s\bigr].
\)
Our target is the CATE in the RCT population, which depends on the contrast 
between the two arm-specific mean outcomes $\mu^r_{1}(\bm x^r)$ and $\mu^r_{-1}(\bm x^r)$. Mathematically, this target RCT CATE is given by
\begin{align}\label{eq:rct-cate}
    \tau^r(\bm x^r)
  \coloneqq \E\bigl[Y(1)-Y(-1)\mid \bX^r=\bm x^r,S=r\bigr]
  = \mu_1^r(\bm x^r) - \mu_{-1}^r(\bm x^r).
\end{align}
This quantity marginalizes the full CATE on \(\bX\) over the unobserved \(\bV\) in the trial. 
Similarly, the RCT propensity score as
{\small\begin{align}\label{eq:rct-propensity}
    \pi^r_a(\bm x^r)
  = \P\{A=a\mid \bX^r=\bm x^r,S=r\}.
\end{align}}
Next, we introduce the \emph{counterfactual mean outcome} (CMO) which is a variance minimizing choice of the augmentation function used to construct a pseudo-outcome \eqref{eq:pseudo-outcome} for CATE estimation through pseudo-outcome regression that we discuss later in \eqref{eq:pseudo_outcome_reg}. The CMO can be defined for any source, but here we focus on the RCT and define
{\small\begin{align} \label{eq:cmo}
    \mu^r(\bm x^r)
   = \sum_{a\in\{-1,1\}} \pi^r_{-a}(\bm x^r)\,\mu^r_a(\bm x^r).
\end{align}}
This quantity is the main building block for R-OSCAR: in our estimator, $\mu^r$ serves
as a personalized baseline outcome that we subtract from $Y$ to construct pseudo-outcomes
for CATE regression \citep{asiaee2023leveraging}. To build intuition, note that for
any fixed individual $\bm x^r$, the CMO is a particular weighted average of the arm-specific
predicted outcomes, where each arm is weighted by the RCT randomization probability
of the opposite arm. Informally, for treatments $1$ and $-1$, the CMO takes the
predicted outcome under treatment $1$ and weights it by the probability of being
assigned $-1$, and vice versa, then sums these two components. This swapped-weight
construction is chosen so that the CMO is highly predictive of the outcome but nearly
uncorrelated with the treatment indicator under the RCT randomization scheme.

Operationally, the CMO acts as a personalized baseline outcome: it captures the part
of an individual's response that can be predicted from covariates alone, before using
the actual treatment they received. In the R-OSCAR framework, one subtracts this
baseline from the observed outcome to form a transformed outcome (a pseudo-outcome)
and then regresses this transformed outcome on covariates. This removes noise due to treatment-free effects of covariates and leaves a cleaner signal for the CATE, while the causal
identification still comes entirely from randomization in the RCT. It should be noted that the OS enters only
through improved prediction of the arm-specific mean outcomes $\{\mu^s_a\}$ and hence
through a better estimate of the CMO. Following \cite{asiaee2023leveraging}, we construct the pseudo-outcomes
\begin{align}
  \tau_{m}(\bX^r,A,Y)
    = \frac{A\{Y-m(\bX^r)\}}{\pi^{r}_{A}(\bX^r)},
    \label{eq:pseudo-outcome}
\end{align}
for a generic augmentation function \(m:\mathbb R^{p_u+p_z}\to\mathbb R\).

A natural way to view this construction is as a specific instance of the broader class of transformed-outcome and pseudo-outcome methods for CATE estimation. These approaches recast CATE estimation as a standard prediction problem by constructing a derived outcome whose conditional mean equals the CATE. For example, \citet{athey2016recursive} and \citet{wager2018estimation} use ``transformed outcomes" and orthogonalized pseudo-outcomes within tree- and forest-based learners so that flexible prediction methods can directly target heterogeneous effects. Building on this idea, \citet{nie2021quasi} and \citet{kennedy2023towards} develop residualized and doubly robust pseudo-outcomes whose conditional expectation is the CATE, yielding estimators with strong robustness and efficiency guarantees while allowing the use of arbitrary machine-learning regressors.

We now view the pseudo-outcome regression as an empirical risk
minimization (ERM) problem in the RCT.  Define the population
squared-error loss
$\mathcal L(f)
    := \E\Bigl\{
           \bigl(\tau_m(\bX^r,A,Y) - f(\bX^r)\bigr)^2 \mid S = r
         \Bigr\},$
where the expectation is over the RCT distribution
$(\bX^r,A,Y)\sim P(\cdot\mid S=r)$.
Under the RCT identification assumptions in
Assumption~\ref{as:ident} stated below, the pseudo-outcome is
conditionally unbiased for the CATE:
\[
  \E\bigl[\tau_m(\bX^r,A,Y)\mid \bX^r, S = r\bigr]
    = \tau^r(\bX^r).
\]
Consequently, $\tau^r$ is the unconstrained minimizer of the
population loss, $\tau^r \in \arg\min_{f} \mathcal L(f)$.
In practice, we restrict optimization to a function class $\cD$ on
$\bX^r$, so the empirical risk minimizer
{\begin{align}\label{eq:pseudo_outcome_reg}
    \hat\tau \;\in\;
  \arg\min_{f\in\cD}\;
    \frac{1}{n^r}\sum_{i:S_i=r}
      \bigl\{\tau_m(\bX_i^r,A_i,Y_i) - f(\bX_i^r)\bigr\}^2
\end{align}}
targets the best approximation to $\tau^r$ within $\cD$, with excess
risk governed by the approximation error
$\Delta_2^2(\cD,\tau^r)$ and the complexity of $\cD$
(see Table~\ref{tab:notation} for formal definitions). Moreover, \cite{yu2021diagnosis} and \cite{asiaee2023leveraging} show that the conditional variance of
\(\tau_m\) given \(\bX^r\) is minimized by \(m=\tilde\mu^r\), where $\tilde\mu^r$ is the marginalized RCT CMO given by {\small\begin{align}
  \tilde \mu^r(\bm x^r)
    &= \E\bigl[\mu^r(\bX)\mid \bX^r=\bm x^r,S=r\bigr].
    \label{eq:rct-cmo}
\end{align}} Thus the CATE 
estimator obtained by regressing pseudo-outcomes on \(\bX^r\) is most 
precise when the augmentation is the marginalized RCT CMO, yielding an augmented inverse probability weighting (AIPW) form \citep{robins1994estimation, zhang2012estimating}. This observation 
motivates using OS data to obtain a better estimator of \(\tilde\mu^r\), while 
keeping the causal identification in the RCT. We impose the following assumptions.
\begin{assumption}[Internal validity of the RCT]\label{as:ident}
SUTVA holds; 
\((Y(1),Y(-1))\ind A\mid (\bX^r,S=r)\) (RCT ignorability);
and there exists \(\rho>0\) such that 
\(0<\rho\le\pi^r_a(\bX^r)\le 1-\rho\) almost surely for each \(a\in\{-1,1\}\).
\end{assumption}
\begin{assumption}[Transportability of $\bV$ given $\bZ$]\label{as:transport}
The conditional distribution of $\bV$ given $\bZ$ is the same in RCT and OS, and 
\(\bV\) is conditionally independent of $\bU$ given $\bZ$ in the RCT:
\[
  P(\bV\mid \bU,\bZ,S=r)
    = P(\bV\mid \bZ,S=r)
    = P(\bV\mid \bZ,S=o).
\]
\end{assumption}
\noindent Assumption~\ref{as:ident} ensures that the pseudo-outcome regression in the RCT 
truly targets the CATE in~\eqref{eq:rct-cate}, so that borrowing from the OS is 
used purely for variance reduction rather than {to correct for unmeasured
confounding or violations of randomization. 
Assumption~\ref{as:transport} formalizes our covariate-mismatch setting:
OS-only covariates \(\bV\) are conditionally transportable given the shared
block \(\bZ\), and \(\bZ\) renders \(\bU\) and \(\bV\) conditionally independent
in the RCT, i.e. \(\bU \perp\!\!\!\perp \bV \mid \bZ, S = r\).
 In particular, this assumption justifies fitting a prediction function 
\(g:\mathbb{R}^{p_z}\to\mathbb{R}^{p_v}\) for \(\bV\) given \(\bZ\) in the OS and 
then using \(\hat g\) to impute \(\widehat\bV = \hat g(\bZ)\) in the RCT. Under 
Assumption~\ref{as:transport}, this imputation does not introduce bias 
in the RCT; the discrepancy between \(\widehat\bV\) and the unobserved \(\bV\) 
manifests only as additional imputation noise, which we track explicitly in our 
risk bounds.}

Throughout, our target of inference is the CATE in the RCT population. The observational source is used purely as an auxiliary data source, not as a second target population. In particular, we \emph{do not} assume that arm-specific outcome regressions agree between the two sources; in general, \(\mu^r_a(\bm x) \neq \mu^o_a(\bm x)\). All causal identification and population interpretation are anchored in the RCT; the OS sample is only used to improve prediction of outcomes or covariates that are missing or sparse in the trial, and all of these OS-based predictions are always calibrated back to the RCT.
Table~\ref{tab:notation} (Supplement) collects the principal symbols used throughout.


\subsubsection{Baseline Methods under Covariate Mismatch}
Under covariate mismatch, a natural baseline is to ignore the OS data entirely
and estimate the CATE using only the RCT. To do so, one can fit arm-wise outcome
models in the RCT,
\[
\hat\mu^r_a(\bm x^r)
  \in \arg\min_{f}
      \frac{1}{n_a^r}\sum_{i:A_i=a,\; S_i=r}
      \bigl(Y_i-f(\bX_i^r)\bigr)^2
      + \mathcal{P}_a^r(f),
\]
with a penalty \(\mathcal{P}_a^r\) controlling the complexity of the arm-specific
regression class.  These nuisance estimates are then plugged into the
counterfactual mean outcome and the pseudo-outcome to construct a purely
RCT-based CATE estimator.  Plugging it into the CMO as 
\(
\htmu^r(\bm x^r)
  = \sum_{a\in\{-1,1\}}
      \pi^r_{-a}(\bm x^r)\,\hat\mu^r_a(\bm x^r),
\) we obtain the RACER estimator $\hat\tau_{\text{RACER}}$ by regressing the pseudo-outcome \eqref{eq:pseudo-outcome} for each RCT unit $i$ on
\(\bX_i^r\) within the RCT.
Under Assumption~\ref{as:ident} and correct specification of the propensity
scores \(\pi^r_a(\bX^r)\) and arm-wise outcome models \(\mu^r_a(\bX^r)\),
the pseudo-outcome satisfies that its expectation given \(\{\bX_i^r,S_i=r\}\) equals \( \tau^r(\bX_i^r)\),
so \(\hat\tau_{\text{RACER}}\) is centered on the target CATE and any remaining
error is due to estimation variance rather than bias. When
\(n^r\) is modest, however, this purely RCT-based strategy can have relatively
high variance, motivating the use of OS data to stabilize the nuisance
estimators and improve CATE precision.

We now introduce a naive borrowing strategy under covariate mismatch that follows the R-OSCAR logic of fitting outcome models in the OS, calibrating 
them to the RCT via discrepancy functions, and then using the calibrated models 
to construct variance-reducing CMOs for pseudo-outcome regression. Here, we can restrict our attention to 
the shared covariates \(\bZ\). We can fit OS outcome models using only the shared covariates $\bZ$, calibrate these models to the RCT using discrepancy functions on \(\bX^r\), and then use the calibrated predictions to construct the CMO and pseudo-outcomes. This borrows outcome information from the large OS without ever knowing about the OS-only block $\bV$. When outcomes depend strongly on $\bV$, however, restricting to $\bZ$ can introduce an irreducible projection error that limits the benefit of borrowing. Mathematically, for each arm \(a\), we fit an OS outcome model 
that ignores \(\bV\):
\begin{align}
\hat\mu^{o,\mathrm{sh}}_{a}(\bm z)
  \in \arg\min_{f}
      \frac{1}{n_a^o}\sum_{i:A_i=a,S=o}\bigl(Y_i-f(\bZ_i)\bigr)^2
      + \mathcal{P}_a^{o,\mathrm{sh}}(f),
\label{eq:z-only-outcome}
\end{align}
where \(\mathcal{P}_a^{o,\mathrm{sh}}(f)\) regularizes the OS model fitted on the shared
space. We then calibrate these predictions to the RCT using a discrepancy function 
on \(\bX^r\):
\begin{align}
\hat{\delta}^{\mathrm{sh}}_{a}(\bm x^r)
  \in \arg\min_{d}
      \frac{1}{n_a^r}\sum_{i:A_i=a, S = r}
      \Bigl\{Y_i - \hat\mu^{o,\mathrm{sh}}_{a}(\bZ_i) - d(\bX_i^r)\Bigr\}^2
      + \mathcal{P}_a^{\mathrm{sh}}(d),
\label{eq:z-only-discrepancy}
\end{align}
with a penalty \(\mathcal{P}_a^{\mathrm{sh}}(d)\) controlling the calibration complexity. Here \(d(\bm x^r)\) is a discrepancy function that corrects for differences
in the relationship between \(\bX^r\) and \(Y\) in the OS versus the RCT:
if the OS outcome model \(\hat\mu^{o,\mathrm{sh}}_{a}\) were perfectly transportable, the
ideal discrepancy would be close to zero. The penalty
\(\mathcal{P}_{\bU\bZ,a}^r(d)\) regularizes the magnitude and complexity of
this correction. Larger penalization shrinks \(d\) toward zero, effectively
borrowing more strength from the OS model, while smaller penalization allows a
larger correction and hence relies more heavily on the RCT data. Thus, when the
OS and RCT share a similar covariate-outcome relationship, the calibration
step remains small and substantial information can be borrowed; when there is
more mismatch, the penalty prevents over-correction and limits the influence of
the OS.
 The calibrated arm-specific prediction on \(\bX^r\) is $\hat\mu^{\mathrm{sh}}_a(\bm x^r)
  = \hat\mu^{o,\mathrm{sh}}_{a}(\bm z) + \hat\delta^{\mathrm{sh}}_{a}(\bm x^r),$
which we aggregate into a shared-only CMO and an induced preliminary CATE respectively as
\[
\htmu^{\mathrm{sh}}(\bm x^r)
  = \sum_{a\in\{-1,1\}} \pi^r_{-a}(\bm x^r)\,
      \hat\mu^{\mathrm{sh}}_a(\bm x^r), \enskip \text{ and } \enskip \httau^{\mathrm{sh}}(\bm x^r)
  = \sum_{a\in\{-1,1\}} a\,\hat\mu^{\mathrm{sh}}_a(\bm x^r).
\]
As in R-OSCAR, we then perform a second-stage CATE calibration on the RCT, 
using the shared-only CMO as augmentation. This second-stage calibration targets potential miscalibration of the
OS-derived CATE itself. The first calibration layer adjusts the OS
outcome model so that, after shrinkage, its arm-specific mean predictions align with the RCT outcomes. However, this shrinkage can still induce bias in the resulting CATE contrasts, since the CATE is a nonlinear functional of the outcome models and may have a different smoothness structure. The second calibration stage therefore directly regresses the pseudo-outcomes on \(\bX^r\), treating the OS-estimated CATE as an initial plug-in estimate and learning an additional correction
term. This separates calibration and regularization for the outcome models from calibration and regularization for the CATE surface, allowing the latter to adapt to its own complexity even when the underlying outcome models are heavily smoothed. We then get 

\begin{align}
\hat \delta^{\mathrm{sh}}
  &= \arg\min_{d \in \mathcal{D}}
      \frac{1}{n^r} \sum_{i: S = r}
      \Biggl\{
        \frac{A_i\bigl(Y_i - \htmu^{\mathrm{sh}}(\bX_i^r)\bigr)}
             {\pi^r_{A_i}(\bX_i^r)}
        - \bigl[\httau^{\mathrm{sh}}(\bX_i^r) + d(\bX_i^r)\bigr]
      \Biggr\}^2 + \mathcal{P}(d),
\label{eq:shared-stage2-calib}
\end{align}
and define the shared-only borrowing SR-OSCAR CATE estimator, writing $\hat\tau_{\text{SR}}$ and $\hat\tau_{\text{MR}}$ for the SR-OSCAR and MR-OSCAR estimators in subscripts, as
\[
\hat{\tau}_{\text{SR}}(\bm x^r)
  = \httau^{\mathrm{sh}}(\bm x^r) + \hat \delta^{\mathrm{sh}}(\bm x^r).
\]
This strategy leverages outcome structure estimated with high precision in the large OS sample and then uses the RCT to correct systematic differences that are explainable by \(\bX^r\). 
It is safe in the sense that it never imputes \(\bV\) in the RCT and therefore 
does not rely on Assumption~\ref{as:transport}. 
However, when outcomes depend strongly on OS-only covariates \(\bV\), the map 
\(\bm z\mapsto\E[\mu^o(\bm z,\bV;a)\mid \bZ=\bm z,S=o]\) may lie far from the chosen 
function class on \(\bZ\) alone, and SR-OSCAR can incur an irreducible 
projection bias that we will quantify in Section~\ref{sec:theory}.

\subsubsection{Proposed Method: MR-OSCAR}

A more robust way to handle covariate mismatch is to be explicit about the fact that the RCT is missing the whole block of covariates $\bV$ and then try to recreate this block in the trial using information from the OS. This mismatch-aware, imputation-augmented borrowing, which we call MR-OSCAR (Mismatch-aware Robust Observational Studies for CMO-Augmented RCT), uses the OS to learn how OS-only covariates 
\(\bV\) relate to the shared block \(\bZ\), then transports the predictable part of 
\(\bV\) into the RCT via imputation. Suppose the imputed proxy for RCT units is $\widehat\bV$. We then fit OS outcome models on \(\bX^o\), calibrate them to the RCT on \((\bX^r,\widehat\bV)\), and use the resulting calibrated predictions to build the CMO and pseudo-outcomes. Under a transportability condition for $\bV$ given $\bZ$, the extra uncertainty
introduced by this step is captured by an explicit imputation error term in our
risk bounds, which we analyze in Section~\ref{sec:theory}. Under
Assumption~\ref{as:transport}, this additional term reflects only imputation
noise rather than bias. Mathematically, we first estimate a mapping \(g:\mathbb R^{p_z}\to\mathbb R^{p_v}\) in the OS:
\[
\hat g
  \in \arg\min_{g}
      \frac{1}{n^o}\sum_{i: S = o}
      \bigl\|\bV_i - g(\bZ_i)\bigr\|^2
      + \mathcal{P}^{\mathrm{im}}(g),
\]
where \(\mathcal{P}(g)\) is a penalty appropriate for the chosen regression class 
(e.g., ridge or LASSO). In other words, \(\hat g\) is trained to approximate the conditional mean 
\(\E(\bV \mid \bZ)\) within the chosen function class, so that for any RCT unit 
with covariates \(\bZ_i\), \(\hat g(\bZ_i)\) provides a regularized prediction of 
its OS-only covariates \(\bV_i\). We then impute $\widehat\bV_i = \hat g(\bZ_i),$ for $i=1,\ldots,n^r,$
for RCT units, and treat \(\widehat\bV\) as a proxy for the predictable component 
of \(\bV\). Next, for each arm \(a\), we fit OS outcome models on the full OS feature set
\(\bX^o\):
\begin{align}
\hat\mu^{o,\mathrm{im}}_{a}(\bm x^o)
  \in \arg\min_{f}
      \frac{1}{n_a^o}\sum_{i:A_i=a, S = o}
      \bigl(Y_i-f(\bX_i^o)\bigr)^2
      + \mathcal{P}_a^{o,\mathrm{im}}(f),
\label{eq:imputation-outcome-model}
\end{align}
and calibrate these models to the RCT in the augmented covariate space 
\((\bX^r,\widehat\bV)\):
\begin{align}
\hat\delta^{\mathrm{im}}_{a}(\bm x^r,\widehat{\bm v})
  \in \arg\min_{d}
      \frac{1}{n_a^r}\sum_{i:A_i=a, S = r}
      \Bigl\{Y_i
            - \hat\mu^{o,\mathrm{im}}_{a}(\bZ_i,\widehat\bV_i)
            - d(\bX_i^r,\widehat\bV_i)
      \Bigr\}^2
      + \mathcal{P}_a^{\mathrm{im}}(d).
\label{eq:imputation-discrepancy}
\end{align}
Without regularization, the discrepancy would simply recover the estimate one would obtain from the RCT outcomes;
the penalty \(\mathcal{P}_a^{\mathrm{im}}(d)\) instead forces the calibration step 
to keep most of the OS prediction and add only a modest RCT correction, allowing for borrowing of information from the OS. The calibrated arm-specific prediction in the augmented space is
\(
\hat\mu^{\mathrm{im}}_a(\bm x^r,\widehat{\bm v})
  = \hat\mu^{o,\mathrm{im}}_{a}(\bm z,\widehat{\bm v})
    + \hat\delta^{\mathrm{im}}_{a}(\bm x^r,\widehat{\bm v}).
\) We define the imputation-augmented CMO and preliminary CATE as
\begin{align}\label{cmo_and_prelim_CATE}
    \htmu^{\mathrm{im}}(\bm x^r)
  = \sum_{a\in\{-1,1\}} \pi^r_{-a}(\bm x^r)\,
      \hat\mu^{\mathrm{im}}_a(\bm x^r,\widehat{\bm v}) \enskip \text{  and  } \enskip \httau^{\mathrm{im}}(\bm x^r)
  = \sum_{a\in\{-1,1\}} a\,
      \hat\mu^{\mathrm{im}}_a(\bm x^r,\widehat{\bm v}).
\end{align}
Finally, we perform the same second-stage CATE calibration as in 
\eqref{eq:shared-stage2-calib}, now using the imputation-augmented CMO:
\begin{align}
\hat \delta^{\mathrm{im}}
  &= \arg\min_{d \in \mathcal{D}}
      \frac{1}{n^r} \sum_{i:S = r}
      \Biggl\{
        \frac{A_i\bigl(Y_i - \htmu^{\mathrm{im}}(\bX_i^r)\bigr)}
             {\pi^r_{A_i}(\bX_i^r)}
        - \bigl[\httau^{\mathrm{im}}(\bX_i^r) + d(\bX_i^r)\bigr]
      \Biggr\}^2 + \mathcal{P}(d),
\label{eq:robust-delta}
\end{align} where, this second-stage calibration takes the OS-derived CATE as an initial plug-in estimate and then fits an additional regression of the pseudo-outcomes on \(\bX^r\) to correct any remaining bias. We define the MR-OSCAR CATE estimator as
{\small\begin{align}\label{MR-OSCAR-estimator}
    \hat \tau_{\text{MR}}(\bm x^r)
  = \httau^{\mathrm{im}}(\bm x^r) + \hat \delta^{\mathrm{im}}(\bm x^r).
\end{align}}
In practice we employ sample-splitting or cross-fitting across the nuisance stages 
(OS outcome modeling, imputation, and RCT calibrations) to avoid overfitting. Intuitively, MR-OSCAR uses a single imputation to expose the predictable 
component of \(\bV\) to the R-OSCAR calibrator, thereby shrinking the residual 
structure that must be learned from the RCT. When \(\bV\) carries substantial 
predictive signal that is partially recoverable from \(\bZ\), we show that MR-OSCAR can reduce 
CATE risk relative to both RACER and SR-OSCAR. When \(\bV\) is weakly 
predictive or poorly imputable, our finite-sample theory in Section~\ref{sec:theory} 
shows that MR-OSCAR effectively falls back toward RCT-only estimation, thereby 
guarding against negative transfer. Recall that we do not impose any transportability of outcome regressions
or CATEs between sources; MR-OSCAR uses the OS only to improve
prediction and then calibrates all borrowed structure back to the RCT.

Table~\ref{tab:cmo-summary} summarizes the estimation pipeline for each method.

\begin{table}[ht]
\centering
\caption{Summary of arm-specific mean, CMO, and CATE estimators across
methods.  A dash (--) indicates the component is not used.}\label{tab:cmo-summary}
\footnotesize
\begin{tabular}{@{}l ccc@{}}
\toprule
\textbf{Component} & \textbf{RACER} & \textbf{SR-OSCAR} & \textbf{MR-OSCAR} \\
\midrule
OS outcome &
  -- &
  $\hat\mu^{o,\mathrm{sh}}_a(\bm z)$ &
  $\hat\mu^{o,\mathrm{im}}_a(\bm x^o)$ \\[4pt]
Imputation &
  -- &
  -- &
  $\widehat\bV = \hat g(\bZ)$ \\[4pt]
Discrepancy &
  -- &
  $\hat\delta^{\mathrm{sh}}_a(\bm x^r)$ &
  $\hat\delta^{\mathrm{im}}_a(\bm x^r,\widehat{\bm v})$ \\[4pt]
\begin{tabular}[c]{@{}l@{}}\scriptsize(Calibrated)\\[-3pt]Arm mean $\hat\mu_a$\end{tabular} &
  $\hat\mu^r_a(\bm x^r)$ &
  $\hat\mu^{\mathrm{sh}}_a\!=\!\hat\mu^{o,\mathrm{sh}}_a\!+\!\hat\delta^{\mathrm{sh}}_a$ &
  $\hat\mu^{\mathrm{im}}_a\!=\!\hat\mu^{o,\mathrm{im}}_a\!+\!\hat\delta^{\mathrm{im}}_a$ \\[4pt]
CMO &
  $\htmu^r\!=\!\textstyle\sum_a \pi^r_{-a}\hat\mu^r_a$ &
  $\htmu^{\mathrm{sh}}\!=\!\textstyle\sum_a \pi^r_{-a}\hat\mu^{\mathrm{sh}}_a$ &
  $\htmu^{\mathrm{im}}\!=\!\textstyle\sum_a \pi^r_{-a}\hat\mu^{\mathrm{im}}_a$ \\[4pt]
Preliminary CATE &
  -- &
  $\httau^{\mathrm{sh}}\!=\!\textstyle\sum_a a\,\hat\mu^{\mathrm{sh}}_a$ &
  $\httau^{\mathrm{im}}\!=\!\textstyle\sum_a a\,\hat\mu^{\mathrm{im}}_a$ \\[4pt]
\begin{tabular}[c]{@{}l@{}}\scriptsize(Calibrated)\\[-3pt]CATE $\hat\tau$\end{tabular} &
  $\hat\tau_{\text{RACER}}(\bm x^r)$ &
  $\hat\tau_{\text{SR}}\!=\!\httau^{\mathrm{sh}}\!+\!\hat\delta^{\mathrm{sh}}$ &
  $\hat\tau_{\text{MR}}\!=\!\httau^{\mathrm{im}}\!+\!\hat\delta^{\mathrm{im}}$ \\
\bottomrule
\end{tabular}
\end{table}
\subsubsection{Example: Sparse Linear Model}
For implementation, we first focus on a sparse linear specification of the
MR-OSCAR nuisance models.
Let \(p_r:=p_u+p_z=\dim(\bX^r)\) and \(p_o:=p_z+p_v=\dim(\bX^o)\). Define the intercept-augmented design vectors
\(\tilde{\bX}^o_i := (1,\bZ_i^\top,\bV_i^\top)^\top \in \mathbb{R}^{p_o + 1}\),
\(\tilde{\bX}^{r,\mathrm{im}}_i := (1,\bU_i^\top,\bZ_i^\top,\widehat\bV_i^\top)^\top \in \mathbb{R}^{p_r+p_v+1}\), and
\(\tilde{\bX}^r_i := (1,\bU_i^\top,\bZ_i^\top)^\top \in \mathbb{R}^{p_r+1}\)
for the OS, the RCT with imputed
$\widehat\bV$, and the RCT with observed $(\bU,\bZ)$, respectively.
We assume that for each arm $a\in\{-1,1\}$, the OS outcome model is $\mu^{o,\mathrm{im}}_{a}(\tilde{\bX}^o_i) = (\tilde{\bm x}^{o}_i)^\top\bm\beta^o_a,
$, where $\bm\beta^o_a$ is $s$-sparse. In this
sparse linear setting, each outcome model and calibration step is a LASSO
regression on the corresponding design matrix, and we estimate the sparse coefficient vector $\bm\beta^o_a$ via
\begin{equation}
\label{eq:os-lasso}
\widehat{\bm\beta}^o_a
  \in \arg\min_{\bm\beta\in\mathbb R^{p_o+1}}
  \biggl\{
    \frac{1}{n_a^o}\sum_{i:A_i=a,S=o}
      \bigl(Y_i - (\tilde{\bX}^o_i)^\top\bm\beta\bigr)^2
    + \lambda_a^o\|\bm\beta\|_1
  \biggr\},
\end{equation}
where $\lambda_a^o\ge 0$ is a LASSO tuning parameter.  The fitted OS
prediction is $\hat\mu^{o,\mathrm{im}}_{a}(\bX_i^o) = (\tilde{\bX}^o_i)^\top\widehat{\bm\beta}^o_a$.
Given the OS fit, we form residuals for RCT units in arm $a$,
$\tilde Y_i(a) = Y_i - \hat\mu^{o,\mathrm{im}}_{a}(\bZ_i,\widehat\bV_i),$ for $i:A_i=a,S=r,$
and model the arm-specific discrepancy as
$d_a(\bX_i^r,\widehat\bV_i) = (\tilde{\bX}^{r,\mathrm{im}}_i)^\top\bm\gamma^r_a$.
We estimate $\bm\gamma^r_a$ via a second LASSO regression,
\begin{equation}
\label{eq:disc-lasso}
\widehat{\bm\gamma}^r_a
  \in \arg\min_{\gamma\in\mathbb R^{p_r+p_v+1}}
  \biggl\{
    \frac{1}{n_a^r}\sum_{i:A_i=a,S=r}
      \bigl(\tilde Y_i(a) - (\tilde{\bX}^{r,\mathrm{im}}_i)^\top\bm\gamma\bigr)^2
    + \lambda_a^r\|\bm\gamma\|_1
  \biggr\}.
\end{equation}
The calibrated arm-specific prediction in the augmented space
$(\bX^r,\widehat\bV)$ is then $\hat\mu^{\mathrm{im}}_a(\bX_i^r,\widehat\bV_i)
  = \hat\mu^{o,\mathrm{im}}_{a}(\bZ_i,\widehat\bV_i)
    + (\tilde{\bX}^{r,\mathrm{im}}_i)^\top\widehat{\bm\gamma}^r_a.$
Using the calibrated predictions, we define the imputation-augmented CMO and
preliminary CATE in the RCT as \eqref{cmo_and_prelim_CATE}. We then form pseudo-outcomes $\psi_i$ as in Equation~\ref{eq:pseudo-outcome} and fit a
final linear correction
$d(\bX_i^r)= (\tilde{\bX}^r_i)^\top\bm\eta$ via
\begin{equation}
\label{eq:delta-lasso}
\widehat{\bm\eta}
  \in \arg\min_{\eta\in\mathbb R^{p_r+1}}
 \biggl\{
    \frac{1}{n^r}\sum_{i:S=r}
      \bigl(\psi_i - [\httau^{\mathrm{im}}(\bX_i^r) +(\tilde{\bX}^r_i)^\top\bm\eta]\bigr)^2
    + \lambda_\delta\|\bm\eta\|_1
  \biggr\}.
\end{equation}
The MR-OSCAR CATE estimator in the sparse linear regime is finally $\hat\tau_{\text{MR}}(\bX^r)
  = \httau^{\mathrm{im}}(\bX^r) + (\tilde{\bX}^r)^\top\widehat{\bm\eta}.$
In practice, we employ
sample-splitting or cross-fitting across
\eqref{eq:os-lasso}-\eqref{eq:delta-lasso} to avoid overfitting. Moreover, the tuning parameters
$\lambda_a^o$, $\lambda_a^r$, and $\lambda_\delta$
are selected separately for each nuisance stage by cross-validation,
using the squared-error criterion associated with that stage's own
regression problem.
Specifically, for the OS arm-specific outcome models in
\eqref{eq:os-lasso}, cross-validation is performed using the original
outcome $Y_i$ as the response and the prediction loss $\sum_{i\in V,\;A_i=a,S=o}
(Y_i-(\tilde{\bX}_i^o)^\top\hat{\bm\beta}_{a,-V}^o)^2$
on each validation fold $V$.
For the RCT arm-specific discrepancy regressions in
\eqref{eq:disc-lasso}, cross-validation uses the residualized outcome $\tilde Y_i(a)=Y_i-\hat\mu_{a}^{o,\mathrm{im}}(\bZ_i,\widehat\bV_i)$ as the response and minimizes the corresponding validation loss $\sum_{i\in V,\;A_i=a,S=r}
(\tilde Y_i(a)-(\tilde{\bX}_i^{r,\mathrm{im}})^\top
\hat{\bm\gamma}_{a,-V}^r)^2.$
Finally, for the CATE correction step in \eqref{eq:delta-lasso},
cross-validation is performed using the pseudo-outcome regression
criterion itself. Equivalently, since
\(\httau^{\mathrm{im}}(\bX_i^r)\) is an offset, one may view the response
as either the pseudo-outcome $\psi_i$ with offset
$\httau^{\mathrm{im}}(\bX_i^r)$, or as the residualized pseudo-outcome
$\psi_i-\httau^{\mathrm{im}}(\bX_i^r)$; both formulations give the same
optimization problem. Thus the validation loss for $\lambda_\delta$ is $\sum_{i\in V,\;S=r}
(\psi_i-[\httau^{\mathrm{im}}(\bX_i^r)
+(\tilde{\bX}_i^r)^\top\hat{\bm\eta}_{-V}])^2.$
Because MR-OSCAR involves multiple nuisance stages (imputation, OS outcome, RCT discrepancy, and CATE regression), we employ $K$-fold cross-fitting: for each fold $I_k$, all nuisance functions---$\hat g^{(-k)}$, $\hat\mu^{o,(-k)}_{a}$, $\hat\delta_a^{(-k)}$, and the preliminary CMO $\htmu^{(-k)}$---are estimated on the complement $I_k^c$, and the final CATE $\hat\tau_{\text{MR}}^{(-k)}$ is evaluated on $I_k$. Aggregating over folds ensures that pseudo-outcomes and CATE regressions use disjoint data, avoiding overfitting.
{
\section{Theory}
\label{sec:theory}
We now study the finite-sample performance of the CATE estimators
\(\hat\tau_{\text{SR}}\) and \(\hat\tau_{\text{MR}}\).
Throughout this section we focus on the population risk in the RCT
distribution, $\Delta_2^2(\hat\tau,\tau^r)
    := \E\bigl\{ \bigl(\hat\tau(\bX^r)-\tau^r(\bX^r)\bigr)^2\mid S = r \bigr\},$
where the expectation is taken over an independent RCT draw
\((\bX^r,A,Y)\sim P(\cdot\mid S=r)\).
Since all of our results are in squared \(L_2\) loss under the RCT
distribution, this notation is meaningful in this context. Our goal is to characterize how much each estimator can reduce this risk
relative to the RCT-only RACER benchmark by borrowing information from
the OS, and how covariate mismatch and imputation error impact the rate.
\subsection{Function classes, complexity measures, and shift structure}
\label{subsec:theory-setup}
We first introduce the function classes and complexity measures that
enter our bounds.
Let \(\cD\subset L_2(P^r)\) be the function class used for the final
CATE calibration \(d\) in \eqref{eq:shared-stage2-calib} and
\eqref{eq:robust-delta}, and for each arm \(a\in\{-1,1\}\) define \(\cM^{o,\mathrm{sh}}_{a}\) and \(\cM^{o,\mathrm{im}}_{a}\) as the function classes used for the OS outcome models
        \(\hat\mu^{o,\mathrm{sh}}_{a}(\bm z)\) in \eqref{eq:z-only-outcome}
        and \(\hat\mu^{o,\mathrm{im}}_{a}(\bm x^o)\) in
        \eqref{eq:imputation-outcome-model}, \(\cM^{r}_{a}\) as the arm-specific RCT outcome classes used
        by RACER on \((\bU,\bZ)=\bX^r\), and \(\cD^{\mathrm{sh}}_{a}\) and \(\cD^{\mathrm{im}}_{a}\) as the calibration
        classes used for \(\hat\delta^{\mathrm{sh}}_{a}\) in
        \eqref{eq:z-only-discrepancy} and
        \eqref{eq:imputation-discrepancy} respectively. Let \(\cG\) be the class used for the imputation map
\(g:\bZ\mapsto\bV\) in MR-OSCAR. We write \(\Rad_n(\cH)\) for the empirical Rademacher
complexity of a class \(\cH\) based on \(n\) samples \citep{bartlett2006local}
(see Table~\ref{tab:notation} for its formal definition), and we use
\(\lesssim\) to suppress universal constants. For any function class \(\cH\subset L_2(P^r)\), define the
{approximation error} as $\Delta_2^2(\cH,\tau^r)
    := \inf_{h\in\cH}
        \E\bigl\{ (h(\bX^r) - \tau^r(\bX^r))^2 \mid S = r \bigr\}.$
This term captures the irreducible bias induced by restricting attention
to a given class (for example, linear or sparse linear functions), while
all remaining terms in our bounds arise from estimation error and from
covariate mismatch. To quantify the effect of ignoring OS-only covariates \(\bV\), we
introduce a shared-only mismatch penalty \eqref{eq:BZ-def}.
For each arm \(a\in\{-1,1\}\), define the arm-specific marginalized RCT mean (the per-arm counterpart of the CMO \eqref{eq:rct-cmo})
\[
  \tilde\mu^r_a(\bm x^r)
    := \E\bigl[
          \mu^r_a(\bm x^r,\bV)
        \mid \bX^r=\bm x^r, S = r
       \bigr].
\]
Under SR-OSCAR, the armwise calibrated predictor has the form \(m(\bZ)+d(\bX^r)\) with \(m\in\cM^{o,\mathrm{sh}}_{a}\) and \(d\in\cD^{\mathrm{sh}}_{a}\). Define the corresponding SR-representable class
\[
  \cH^{\mathrm{sh}}_{a}
    := \Bigl\{
          \bm x^r\mapsto m(\bm z)+d(\bm x^r)
          : m\in\cM^{o,\mathrm{sh}}_{a},\ d\in\cD^{\mathrm{sh}}_{a}
       \Bigr\}.
\]
We define the squared mismatch penalty for shared-only borrowing as the residual approximation error of \(\tilde\mu^r_a\) in \(\cH^{\mathrm{sh}}_{a}\):
\begin{equation}
  B_{\mathrm{sh}}^2
    := \sum_{a\in\{-1,1\}}
        \inf_{h\in\cH^{\mathrm{sh}}_{a}}
        \E\Bigl[
          \bigl\{
            h(\bX^r)
            - \tilde\mu^r_a(\bX^r)
          \bigr\}^2 \mid S = r
        \Bigr].
  \label{eq:BZ-def}
\end{equation}
Intuitively, \(B_{\mathrm{sh}}^2\) is the irreducible penalty from restricting SR-OSCAR’s armwise calibration to functions of \(\bX^r\): it is small when OS-only covariates \(\bV\) do not materially affect outcomes beyond what \(\bX^r\) can represent, and it is large when important effect modifiers reside in \(\bV\) and cannot be summarized by \(\bX^r\).
For MR-OSCAR, we track the quality of imputation via
\[
  r_{\mathrm{im}}^2
    := \E\bigl[\|\bV - g^\star(\bZ)\|^2 \mid S = r\bigr],
\]
where \(g^\star\in\cG\) denotes an ideal imputation map in the chosen
class and the expectation is taken under the RCT distribution.
In practice \(r_{\mathrm{im}}^2\) can be thought of as the prediction risk of
the best imputation model.
One may also consider the OS-averaged version \(\E[\|\bV - g^\star(\bZ)\|^2 \mid S=o]\); under transportability of \(\bV\mid\bZ\) and mild overlap between \(P(\bZ\mid S=r)\) and \(P(\bZ\mid S=o)\), the OS- and RCT-averaged imputation risks are of the same order.

 
Finally, for some parts of the theory it is convenient to use
{localized} complexity parameters in the sense of
\citet{asiaee2023leveraging} following \cite{bartlett2006local}.
Given a function class \(\cH\) and rate exponent \(\eta>0\), we write
\(c(\cH)\) for a localized complexity of \(\cH\) such that
\begin{equation}
  \Delta_2^2(\hat h_n,h^\star)
    \;\lesssim\;
    \frac{c(\cH)}{n^{\eta}} \text{ for ERM in \(\cH\) based on \(n\) samples,}
  \label{eq:localized-generic}
\end{equation}
under correct specification \(h^\star\in\cH\) and standard regularity.
For instance, for sparse linear models (LASSO) with sparsity \(s\) among
\(p\) features one has \(c(\cH)\asymp s\log p\) and \(\eta=1\); for
Hölder-smooth nonparametric classes one obtains
\(\eta=2\alpha/(2\alpha+p)\), etc.
The Rademacher-based statements below and the localized forms
\eqref{eq:localized-generic} are equivalent up to constants. The following assumptions formalize the outcome-shift structure and
smoothness conditions needed for our analysis.

\begin{assumption}[Outcome shift in the full space]
\label{as:shift}
For each arm \(a\in\{-1,1\}\) there exists a function
\(\delta_a:\R^{p}\to\R\) such that $\mu^r_a(\bm x)
    = \mu^o_a(\bm x) + \delta_a(\bm x),$
and \(\delta_a\) belongs to a function class \(\cD_a\) of controlled
complexity (e.g., having bounded Rademacher complexity).
\end{assumption}

\begin{assumption}[Lipschitzness in the OS-only block]
\label{as:lipschitz}
There exists \(L>0\) such that for all $ s\in\{r,o\}$ and for all 
\((\bm x^r,\bm v,\bm v',a,s)\), we have $\bigl|
    \mu^s_a(\bm x^r,\bm v) - \mu^s_a(\bm x^r,\bm v')
  \bigr|
  \le L \|\bm v - \bm v'\|
  \text{ and } \bigl|
    \delta_a(\bm x^r,\bm v) - \delta_a(\bm x^r,\bm v')
  \bigr|
  \le L \|\bm v - \bm v'\|.$
\end{assumption}

\begin{assumption}[Rates for nuisance estimation]
\label{as:rates}
The estimators
\(\hat\mu^{o,\mathrm{sh}}_{a}\), \(\hat\mu^{o,\mathrm{im}}_{a}\),
\(\hat\delta^{\mathrm{sh}}_{a}\), \(\hat\delta^{\mathrm{im}}_{a}\), \(\hat g\), and the final-stage calibrators
\(\hat\delta^{\mathrm{sh}}\), \(\hat\delta^{\mathrm{im}}\) are obtained by empirical risk
minimization (or penalized ERM) in their respective function classes,
with sample splitting or cross-fitting across nuisance stages.
Their estimation errors admit bounds of the form $\E\Bigl[
    \bigl\|\hat h - h^\star\bigr\|_{2}^2
	  \Bigr]
	  \;\lesssim\;
	  \Rad_{n}^2(\cH) \text{ for each nuisance component } h^\star\in\cH,$
	with \(n=n^r\) or \(n^o\) as appropriate.
\end{assumption}

Assumption~\ref{as:shift} allows the arm-specific mean outcomes to
differ arbitrarily between the RCT and OS, encoded by the shift
functions \(\delta_a\).
Crucially, we do {not} assume that the outcome means or CATEs
transport between populations in the sense that
we allow \(\mu^r_a\ne\mu^o_a\) and \(\tau^r\ne\tau^o\).
The OS is used only as a high-quality predictor of outcomes, with the
shift structure captured and corrected on the RCT via the calibration
steps.
Assumption~\ref{as:lipschitz} ensures that small changes in the OS-only
covariates \(\bV\) lead to small, predictable changes in both the
outcome model and the discrepancy function, which is crucial when we
impute \(\bV\) into the trial.
Assumption~\ref{as:rates} (combined with
\eqref{eq:localized-generic}) ensures that the nuisance estimators are
sufficiently accurate for standard Rademacher/localized-complexity
arguments to apply.

\subsection{Baseline risk bound for RACER}
\label{subsec:racer-theory}

As a benchmark, we recall the risk bound for the RCT-only estimator
RACER in our current notation.
RACER fits arm-specific outcome models
\(\hat\mu^r_{a}(\bX^r)\) in the RCT only and then constructs CMOs and
pseudo-outcomes using these models, followed by a final CATE calibration
in the class \(\cD\).
The following result is a straightforward adaptation of the main
R-OSCAR risk bound in \citet{asiaee2023leveraging} to the covariate
block \(\bX^r\).

\begin{theorem}[Baseline risk bound for RACER, adapted from \citet{asiaee2023leveraging}]
\label{thm:racer}
Suppose Assumptions~\ref{as:ident} and \ref{as:rates} hold, and that
RACER uses arm-specific outcome classes \(\cM^{r}_{a}\) and final CATE
class \(\cD\), with nuisance estimators obtained via cross-fitting.
Then there exists a constant \(C>0\) such that, with probability at
least \(1-\gamma\) for all \(\gamma\in(0,1)\), $\Delta_2^2(\hat\tau_{\mathrm{RACER}},\tau^{r})$ can be bounded by
\begin{align} \label{eq:racer-bound}
\Delta_2^2(\cD,\tau^{r})
+
C\Biggl[
    \Rad_{n^r}^2(\cD)
    + \sum_{a\in\{-1,1\}}
        \Rad_{n^r}^2\bigl(\cM^{r}_{a}\bigr)
\Biggr]
+ C\,\frac{\log(1/\gamma)}{n^r}.
\end{align}
Equivalently, in the localized-complexity notation of
\citet{asiaee2023leveraging}, if \(c(\cM^{r}_{a})\) denotes a
localized complexity of the arm-specific RCT outcome class on
\(\bX^r\), then under the same conditions there exists \(C<\infty\)
and \(\eta_r>0\) such that
\begin{align}
  \Delta_2^2\bigl(\hat\tau_{\mathrm{RACER}},\tau^{r}\bigr)
    \;\lesssim\;
    \sum_{a\in\{-1,1\}}
      \frac{c\bigl(\cM^{r}_{a}\bigr)}{(n^r)^{\eta_r}}.
  \label{eq:racer-localized}
\end{align}
\end{theorem}
\noindent The bound~\eqref{eq:racer-bound} will serve as our baseline for
comparison with SR-OSCAR and MR-OSCAR: all three estimators share
the same approximation error \(\Delta_2^2(\cD,\tau^{r})\), but RACER
does not incur mismatch or imputation penalties and relies solely on
RCT-based outcome models whose complexity scales with \(n^r\).

\subsection{Error bounds for MR-OSCAR and imputation penalties}
\label{subsec:theory-mroscar}

We now turn to the mismatch-aware estimator MR-OSCAR, which augments
the RCT covariates with imputed OS-only covariates \(\widehat\bV\).
Under Assumption~\ref{as:transport}, the distribution of \(\bV\) given
\(\bZ\) is transportable across RCT and OS, and \(\bZ\) screens off the
dependence between \(\bU\) and \(\bV\) in the RCT.
Combined with the Lipschitz condition in
Assumption~\ref{as:lipschitz}, this lets us relate the error due to
imputation to the prediction error of the outcome models. Let \(\hat\tau_{\text{MR}}\) denote the CATE estimator constructed in \eqref{MR-OSCAR-estimator}, using an imputation class \(\cG\) with
oracle risk \(r_{\mathrm{im}}^2\) and final calibration class \(\cD\).
Let \(m_{\mathrm{aug},a}(\bX^r)\) denote the arm-specific augmentation
function implicitly defined by MR-OSCAR after imputation and
calibration, and define the corresponding augmentation error
\[
  \Delta_2^2(m_{\mathrm{aug},a},\tilde\mu^r_a\mid\bX^r)
    := \E\Bigl[
          \bigl\{
            m_{\mathrm{aug},a}(\bX^r)
            - \tilde\mu^r_a(\bX^r)
          \bigr\}^2 \mid S = r
        \Bigr].
\]

\begin{theorem}[Risk bound for MR-OSCAR]
\label{thm:mroscar}
Suppose Assumptions~\ref{as:ident}-\ref{as:transport} and
\ref{as:shift}-\ref{as:rates} hold, and the nuisances are estimated
with cross-fitting.
Then there exists a constant \(C>0\) such that, with probability at
least \(1-\gamma\) for all \(\gamma\in(0,1)\), we can bound $\Delta_2^2(\hat\tau_{\text{MR}},\tau^r)$ by 
{\small\begin{align}
\Delta_2^2(\cD,\tau^r)
+
C\Biggl[
    L^2 r_{\mathrm{im}}^2
    + \Rad_{n^r}^2(\cD)
    + \sum_{a\in\{-1,1\}}
      \Bigl(
        \Rad_{n^o}^2(\cM^{o,\mathrm{im}}_{a})
        +
        \Rad_{n^r}^2(\cD^{\mathrm{im}}_{a})
      \Bigr)
    + \Rad_{n^o}^2(\cG)
\Biggr]
+ C\,\frac{\log(1/\gamma)}{n^r}.
\label{eq:mroscar-bound}
\end{align}}
\end{theorem}

The bound~\eqref{eq:mroscar-bound} closely parallels the SR-OSCAR bound
\eqref{eq:sroscar-bound} derived in the supplement, but with three key
differences. First, whereas the SR-OSCAR bound contains the
shared-only mismatch penalty $B_{\mathrm{sh}}^2$ arising from discarding $\bV$,
the MR-OSCAR bound replaces this structural penalty by an explicit
imputation term. In particular, the Lipschitz-imputation factor
$L^2 r_{\mathrm{im}}^2$ quantifies the additional error incurred by working
with imputed or predicted OS-only covariates $\widehat{\bV}$ rather than observing
$\bV$ directly. When $\bV$ is highly predictable from $\bZ$, the
oracle imputation risk $r_{\mathrm{im}}^2$ is small and the MR-OSCAR bound
approaches the ideal case in which the full covariate vector is
available in the RCT. Second, MR-OSCAR operates in the full
$\bX^o$ space, so the OS outcome and RCT discrepancy complexities
enter through $\Rad_{n^o}^2(\cM^{o,\mathrm{im}}_{a})$ and
$\Rad_{n^r}^2(\cD^{\mathrm{im}}_{a})$, reflecting the cost of
using richer models that exploit $\bV$. Finally, the complexity of
the imputation class $\cG$ appears via $\Rad_{n^o}^2(\cG)$,
capturing the estimation error in fitting the map
$g:\bZ \mapsto \widehat{\bV}$.  Together, these terms make explicit
how MR-OSCAR trades the shared-only mismatch penalty in SR-OSCAR for a
combination of imputation error and additional modeling flexibility in
the enlarged covariate space.

Together, \Cref{thm:mroscar} and Theorem~\ref{thm:sroscar} in the Supplement provide
transparent conditions under which borrowing from the OS is beneficial.
For example, MR-OSCAR improves on SR-OSCAR when (i) the OS-only
covariates \(\bV\) carry substantial predictive signal for the outcome;
(ii) \(\bV\) can be imputed from \(\bZ\) with small oracle risk
\(r_{\mathrm{im}}^2\); and (iii) the classes \(\cM^{o,\mathrm{im}}_{a}\),
\(\cD^{\mathrm{im}}_{a}\), and \(\cG\) have manageable complexity so that
their Rademacher penalties decay reasonably with \(n^o\) and \(n^r\). At the same time, the bound \eqref{eq:mroscar-bound} makes explicit how
negative transfer is controlled.
If \(\bV\) is poorly predictable from \(\bZ\) (so that \(r_{\mathrm{im}}^2\) is
large) or if the outcome models on \(\bX^o\) are highly complex, the
additional penalties in \eqref{eq:mroscar-bound} may outweigh the gains
from accessing \(\bV\).
In such arrangements, one can tune the penalties in
\eqref{eq:imputation-outcome-model}-\eqref{eq:imputation-discrepancy}
to shrink the influence of OS-only covariates, effectively reverting
MR-OSCAR toward RACER and recovering the robust RCT-only behavior.
Thus the theory highlights how MR-OSCAR can be deployed in a
``safe borrowing'' mode: exploiting informative \(\bV\) when imputation
is reliable, while automatically dampening their effect when it is not.

\paragraph{Augmentation-error decomposition for MR-OSCAR.}

Analogously to \eqref{eq:sroscar-aug-layer}, the analysis of
MR-OSCAR yields an intermediate bound in terms of augmentation errors:
\[
  \Delta_2^2(\hat\tau_{\text{MR}},\tau^r)
  \;\lesssim\;
  \Delta_2^2(\cD,\tau^r)
  +
  \Biggl(
	    1 + \sum_{a\in\{-1,1\}}
	            \Delta_2^2(m_{\mathrm{aug},a},\tilde\mu^r_a\mid\bX^r)
	  \Biggr)\Rad_{n^r}(\cD),
\]
up to negligible residual terms.
The next proposition shows how the augmentation errors decompose into
complexity and imputation terms.

\begin{proposition}[Augmentation error for MR-OSCAR]
\label{prop:mroscar-aug}
Under the conditions of Theorem~\ref{thm:mroscar}, there exist
constants \(C<\infty\), \(\eta_o>0\), and \(\eta_r>0\) such that, for
each arm \(a\),
\begin{equation}
  \Delta_2^2(m_{\mathrm{aug},a},\tilde\mu^r_a\mid\bX^r)
  \;\lesssim\;
  \frac{c(\cM^{o,\mathrm{im}}_{a})}{(n^o)^{\eta_o}}
  +
  \frac{c(\cD^{\mathrm{im}}_{a})}{(n^r)^{\eta_r}}
  +
  L^2 r_{\mathrm{im}}^2.
  \label{eq:mroscar-aug-bound}
\end{equation}
\end{proposition}
\noindent Combining this with the generic localized bound
\eqref{eq:localized-generic} yields a localized version of
\eqref{eq:mroscar-bound}:
\begin{equation}
  \Delta_2^2(\hat\tau_{\text{MR}},\tau^r)
  \;\lesssim\;
  \Delta_2^2(\cD,\tau^r)
  +
  L^2 r_{\mathrm{im}}^2
  +
  \frac{c(\cD)}{(n^r)^{\eta_r}}
  +
  \sum_{a\in\{-1,1\}}
  \biggl\{
    \frac{c(\cM^{o,\mathrm{im}}_{a})}{(n^o)^{\eta_o}}
    +
    \frac{c(\cD^{\mathrm{im}}_{a})}{(n^r)^{\eta_r}}
  \biggr\}
  +
  \frac{c(\cG)}{(n^o)^{\eta_o}}.
  \label{eq:mroscar-localized}
\end{equation}
\noindent Similar bounds for SR-OSCAR have been derived in the Supplement. The localized bounds \eqref{eq:sroscar-localized} and
\eqref{eq:mroscar-localized} also yield conditions under which
MR-OSCAR can strictly dominate SR-OSCAR. MR-OSCAR improves on SR-OSCAR when the additional
complexity of modeling \(\bX^o\) and the imputation error \(r_{\mathrm{im}}^2\)
are more than offset by the reduction in mismatch penalty obtained from
accessing \(\bV\). This is concretely formalized in the following Corollary. 

\begin{corollary}[MR-OSCAR can dominate SR-OSCAR]
\label{cor:mroscar-vs-sroscar}
Under the assumptions of Theorems~\ref{thm:sroscar} in the Supplement
and~\Cref{thm:mroscar} with the same CATE class \(\cD\) and RCT size
\(n^r\), we have that for all sufficiently large \(n^r,n^o\), $\Delta_2^2(\hat\tau_{\text{MR}},\tau^r)
  \;\le\;
  \Delta_2^2(\hat\tau_{\text{SR}},\tau^r)$
whenever
\begin{align*}
\sum_{a\in\{-1,1\}}\!\Bigg\{
  \frac{c(\cM^{o,\mathrm{im}}_{a})}{(n^o)^{\eta_o}}
  +\frac{c(\cD^{\mathrm{im}}_{a})}{(n^r)^{\eta_r}}
  +L^2 r_{\mathrm{im}}^2
\Bigg\} + \frac{c(\cG)}{(n^o)^{\eta_o}}
<
\sum_{a\in\{-1,1\}}\!\Bigg\{
  \frac{c(\cM^{o,\mathrm{sh}}_{a})}{(n^o)^{\eta_o}}
  +\frac{c(\cD^{\mathrm{sh}}_{a})}{(n^r)^{\eta_r}}
  +B_{\mathrm{sh}}^2
\Bigg\}.
\end{align*}
\end{corollary}

\subsection{Specialization of risk bounds to sparse linear models}
\label{subsec:linear-example}

The general bounds above apply to a wide range of nonparametric
learners.
To provide more concrete guidance, we now specialize to sparse linear
models, which yield explicit sample-size and signal-strength thresholds. Assume that all covariates are centered and bounded, and that the true
CATE belongs to a sparse linear class,
\[
  \tau^{r}(\bm x^r)
    = \bm\beta^{\top}\bm x^r,
  \qquad \|\bm\beta\|_0 \le s_\tau,
\]
for some sparsity level \(s_\tau\). Similarly, the OS outcome models,
discrepancy models, and imputation map are assumed to lie in sparse
linear classes over their respective covariates. Because the SR- and MR-OSCAR
pipelines operate on different covariate spaces, each class carries its own
sparsity:
\(\cM^{o,\mathrm{sh}}_{a}\) on \(\bZ\in\R^{p_z}\)
with sparsity \(s_\mu^{\mathrm{sh}}\),
\(\cM^{o,\mathrm{im}}_{a}\) on \(\bX^o\in\R^{p_o}\)
with sparsity \(s_\mu^{\mathrm{im}}\),
\(\cD^{\mathrm{sh}}_{a}\) on \(\bX^r\in\R^{p_r}\)
with sparsity \(s_\delta^{\mathrm{sh}}\),
\(\cD^{\mathrm{im}}_{a}\) on \((\bX^r,\widehat\bV)\in\R^{p}\)
with sparsity \(s_\delta^{\mathrm{im}}\), and
\(\cG\) on \(\bZ\in\R^{p_z}\).
Since \(\bZ\subset\bX^o\) and \(\bX^r\subset(\bX^r,\widehat\bV)\), any
sparse model on the smaller space embeds into the larger one (by zeroing
the extra coordinates), so
\(s_\mu^{\mathrm{sh}}\le s_\mu^{\mathrm{im}}\) and
\(s_\delta^{\mathrm{sh}}\le s_\delta^{\mathrm{im}}\).
The imputation map \(g:\bZ\to\bV\) predicts each coordinate
\(V_j\) independently from \(\bZ\) via a sparse regression with
sparsity \(s_j:=\|\lambda_{o,j}\|_0\), where \(\lambda_{o,j}\) is the
\(j\)-th row of the coefficient matrix
\(\bm\Lambda_o\in\R^{p_v\times p_z}\).
Penalties are chosen to implement (group) LASSO-type estimators. Standard high-dimensional linear theory then gives
\[
  \Rad_{n^r}^2(\cD)
    \lesssim \frac{s_\tau\log p_r}{n^r},
  \quad
  \Rad_{n^o}^2(\cM^{o,\mathrm{sh}}_{a})
    \lesssim \frac{s_\mu^{\mathrm{sh}}\log p_z}{n^o},
  \quad
  \Rad_{n^o}^2(\cM^{o,\mathrm{im}}_{a})
    \lesssim \frac{s_\mu^{\mathrm{im}}\log p_o}{n^o},
\]
\[
  \Rad_{n^r}^2(\cD^{\mathrm{sh}}_{a})
    \lesssim \frac{s_\delta^{\mathrm{sh}}\log p_r}{n^r},
  \quad
  \Rad_{n^r}^2(\cD^{\mathrm{im}}_{a})
    \lesssim \frac{s_\delta^{\mathrm{im}}\log p}{n^r},
  \quad
  \Rad_{n^o}^2(\cG)
    \lesssim \frac{(\sum_j s_j)\log p_z}{n^o}.
\]
Plugging these into
\eqref{eq:sroscar-bound}--\eqref{eq:mroscar-bound} yields
\[
  \Delta_2^2\bigl(\hat\tau_{\text{SR}},\tau^r\bigr)
    \lesssim
      \Delta_2^2(\cD,\tau^r)
      +
      B_{\mathrm{sh}}^2
      +
      \frac{(s_\tau+s_\delta^{\mathrm{sh}})\log p_r}{n^r}
      +
      \frac{s_\mu^{\mathrm{sh}}\log p_z}{n^o},
\]
and
\[
  \Delta_2^2\bigl(\hat\tau_{\text{MR}},\tau^r\bigr)
    \lesssim
      \Delta_2^2(\cD,\tau^r)
      +
      L^2 r_{\mathrm{im}}^2
      +
      \frac{s_\tau\log p_r}{n^r}
      +
      \frac{s_\delta^{\mathrm{im}}\log p}{n^r}
      +
      \frac{s_\mu^{\mathrm{im}}\log p_o}{n^o}
      +
      \frac{(\sum_j s_j)\log p_z}{n^o},
\]
up to constants and lower-order logarithmic factors.
The ordering \(s_\mu^{\mathrm{sh}}\le s_\mu^{\mathrm{im}}\) and
\(s_\delta^{\mathrm{sh}}\le s_\delta^{\mathrm{im}}\) together with
\(p_z\le p_o\) and \(p_r\le p\) show that MR-OSCAR incurs a strictly
larger statistical cost than SR-OSCAR in every Rademacher term, and
it pays an additional imputation estimation cost
\((\sum_j s_j)\log p_z/n^o\) that SR-OSCAR avoids entirely.
MR-OSCAR's sole advantage is replacing the shared-only mismatch penalty
\(B_{\mathrm{sh}}^2\) with the imputation risk \(L^2 r_{\mathrm{im}}^2\),
which can be substantially smaller when \(\bV\) is predictable from
\(\bZ\). Because the extra statistical and imputation costs are all
\(O(1/n^o)\) or \(O(1/n^r)\), they vanish with sample size, whereas
\(B_{\mathrm{sh}}^2\) and \(r_{\mathrm{im}}^2\) are population-level
quantities that persist regardless of sample size. Thus,
with sufficiently large \(n^o\) and \(n^r\), the comparison reduces
to \(B_{\mathrm{sh}}^2\) versus \(L^2 r_{\mathrm{im}}^2\).
When \(\bV\) is weakly predictive of \(Y\) so that including it in
the OS outcome model does not meaningfully reduce the mismatch
\(B_{\mathrm{sh}}^2\), then \(L^2 r_{\mathrm{im}}^2\) is comparable to
or larger than \(B_{\mathrm{sh}}^2\), and the extra statistical costs of
MR-OSCAR offer no compensating gain; in this regime SR-OSCAR is
preferable.
When \(\bV\) is difficult to impute from \(\bZ\) (low
\(R^2(\bV\!\mid\!\bZ)\)), the imputation risk \(r_{\mathrm{im}}^2\) is
large and introduces noise that further inflates the MR-OSCAR bound,
again favoring SR-OSCAR or even the RCT-only RACER.

To further interpret \(r_{\mathrm{im}}^2\), consider the linear-Gaussian
imputation model $\bV = \bm\Lambda_s \bZ + \varepsilon_s, s\in\{o,r\},$ with \(\varepsilon_s\!\perp\!\bZ\),
\(\mathrm{Cov}(\varepsilon_s)=\bm\Sigma_{\bV\mid \bZ}^{\,s}\), and
\(\mathrm{Cov}_s(\bZ)=\bm\Sigma_{\bZ\bZ}^{\,s}\).
In the population (oracle) case where the true map \(\bm\Lambda_o\) is
known and used for imputation in both sources, one obtains $r_{\mathrm{im}}^2
  =
  \tr\!\bigl(
    (\bm\Lambda_o-\bm\Lambda_r)\,
    \bm\Sigma_{\bZ\bZ}^{\,r}\,
    (\bm\Lambda_o-\bm\Lambda_r)^\top
  \bigr)
  +
  \tr\!\bigl(\bm\Sigma_{\bV\mid \bZ}^{\,r}\bigr),$
which cleanly separates an OS\(\to\)RCT mean-relation shift term and an
irreducible RCT noise term. Under Assumption~\ref{as:transport}, the conditional law \(P(\bV\mid\bZ,S=o)=P(\bV\mid\bZ,S=r)\) implies \(\bm\Lambda_o=\bm\Lambda_r\) in this linear model (so the shift trace term is \(0\)), and the oracle imputation risk reduces to \(r_{\mathrm{im}}^2=\tr\!\bigl(\bm\Sigma_{\bV\mid \bZ}^{\,r}\bigr)\). If instead \(\bm\Lambda_o\) is learned in the OS via a row-sparse LASSO
with per-coordinate sparsity \(s_j:=\|\lambda_{o,j}\|_0\) (as introduced above), then one can establish the bound stated in the following theorem.

\begin{theorem}[MR-OSCAR risk bound in the sparse linear setting]\label{thm:MR-linear}
Assume $\bX$ is sub-Gaussian and the population Gram matrices on the relevant supports satisfy restricted eigenvalue (RE) conditions \citep{van2009conditions}. Under Assumptions \ref{as:transport} and \ref{as:lipschitz}, the imputation error satisfies
\[
r_{\mathrm{im}}^2 \lesssim
\mathrm{tr}\!\bigl((\bm\Lambda_o^\star-\bm\Lambda_r)\,\bm\Sigma_{\bZ\bZ}^{\,r}\,(\bm\Lambda_o^\star-\bm\Lambda_r)^\top\bigr)
+
\kappa\,
\frac{\log p_z}{n^o}\,\sum_{j=1}^{p_v}s_j\,\sigma_{o,j}^2
+
\mathrm{tr}\!\bigl(\bm\Sigma_{\bV\mid \bZ}^{\,r}\bigr),
\]
where $\sigma_{o,j}^2:=(\bm\Sigma_{\bV\mid \bZ}^{\,o})_{jj}$ is the
$j$-th conditional variance,
$\bm\Lambda_o^\star$ is the best row-sparse approximation to $\bm\Lambda_o$, and
$\kappa=\bigl\|\bm\Sigma_{\bZ\bZ}^{\,r}(\bm\Sigma_{\bZ\bZ}^{\,o})^{-1}\bigr\|_{\mathrm{op}}$.

\end{theorem}
Thus the only extra price of bringing in \(\bV\) is the imputation error
\(r_{\mathrm{im}}^2\), which is small precisely when \(\bV\) is predictable from
\(\bZ\) (high \(R^2(\bV\!\mid\!\bZ)\)) and the OS\(\to\)RCT map for
\(\bV\!\mid\!\bZ\) is stable. Altogether, when \(n^o\) is large, \(R^2(\bV\!\mid\!\bZ)\) is moderate
to high, and the calibration class on \(\bX\) is reasonably
sparse, the augmentation-driven factor in the MR-OSCAR bound becomes
small, yielding strictly tighter rates than the RCT-only RACER and the
shared-only SR-OSCAR.
Conversely, if \(\bV\!\mid\!\bZ\) is poorly predictable or undergoes
strong cross-source shifts in conditional distribution, the \(r_{\mathrm{im}}^2\) term can dominate,
warning that mismatch-aware borrowing may not improve upon the simpler baseline estimators, exactly the trade-off the theory is designed to make
explicit.

}

\section{Finite sample experiments}\label{sec:simu}
We consider two data sources with partially overlapping covariates:
the randomized trial observes \(\bX^r=(\bU,\bZ)\) and the observational study observes \(\bX^o=(\bZ,\bV)\).
The total dimension is fixed at $p=100$, with the block dimensions $p_u$, $p_v$, and $p_z=p-p_u-p_v$.
Covariates are generated from an equicorrelated Gaussian distribution
($\rho=0.4$) with a mean shift applied to $\bU$ in the RCT to induce mild domain shift.
The observational sample size is $n^o=1000$ and the RCT size is $n^r=300$, unless mentioned otherwise.
Treatment is assigned at random in the RCT ($\Pr(A=1)=0.5$) and by a logistic model in the OS
using a subset of $\bZ$ so that the treated proportion is approximately $1/3$. Potential outcomes are linear in the covariates with sparse coefficients:
a proportion ($=0.05$) of the \(\bX^r\) coordinates carry signal of
magnitude about $2/3$, and the $\bV$ coefficients have magnitude about $1$; a small fraction
($=0.02$) of the \(\bX^r\) coefficients are perturbed in the RCT to induce
an outcome shift. The outcome depends sparsely on all of $\bX=(\bU,\bZ,\bV)$, but we observe only $\bX^r$ in the RCT and $\bX^o$ in the OS. Homoscedastic Gaussian noise is added to both arms. More details can be found in the supplement. We compare three RCT-anchored CATE procedures:
{RACER}, which uses only RCT data with predictors \(\bX^r\);
{SR-OSCAR}, which borrows from OS using only the shared covariates $\bZ$ on both sources; and
{MR-OSCAR}, which augments RCT predictors with an imputed
$\widehat{\bV}=\hat g(\bZ)$, where the imputation map $g:\mathbb{R}^{p_z}\to\mathbb{R}^{p_v}$ is estimated in the OS and applied to the RCT; OS outcome predictions are then calibrated to the RCT. All OS-anchored procedures are implemented with sample splitting to avoid overfitting.

\subsection{Varying the Effect of OS-only Covariates}
We fix $(p_u,p_z,p_v)=(30,40,30)$ and vary the signal strength of $\bV$ in the outcome model from weak ($0.3$) to strong ($1.3$). The left panel of \Cref{fig:Cov_eff} shows that RMSE increases with the $\bV$ effect for all methods. Across the entire grid, MR-OSCAR has the lowest RMSE, RACER is intermediate, and SR-OSCAR is uniformly worst. MR-OSCAR improves over RACER by $0.03$--$0.10$ RMSE units, with the separation widest at mid-range effects ($\approx0.6$--$1.0$). As the $\bV$ signal strengthens, SR-OSCAR degrades while MR-OSCAR maintains both lower error and greater stability by leveraging OS information.

\subsection{Varying the RCT Sample Size}\label{subsec:vary-nr}
We vary $n^r$ from 200 to 1000 with all else fixed. The right panel of \Cref{fig:Cov_eff} shows three patterns. First, RMSE decreases for all methods as $n^r$ grows. Second, MR-OSCAR attains the lowest RMSE throughout, with the advantage most pronounced at small-to-moderate trial sizes and gradually attenuating as $n^r$ increases, consistent with the $(n^r)^{-1/2}$ scaling of the complexity term. At sufficiently large $n^r$, the RACER and MR-OSCAR curves essentially superimpose. Third, SR-OSCAR is uniformly worst due to the irreducible projection error from restricting to the shared block $\bZ$. In summary, MR-OSCAR is most beneficial when RCT enrollment is limited; as $n^r$ increases the gains diminish but do not reverse.
\begin{figure}[]
    \centering
    \includegraphics[width=0.9\textwidth]{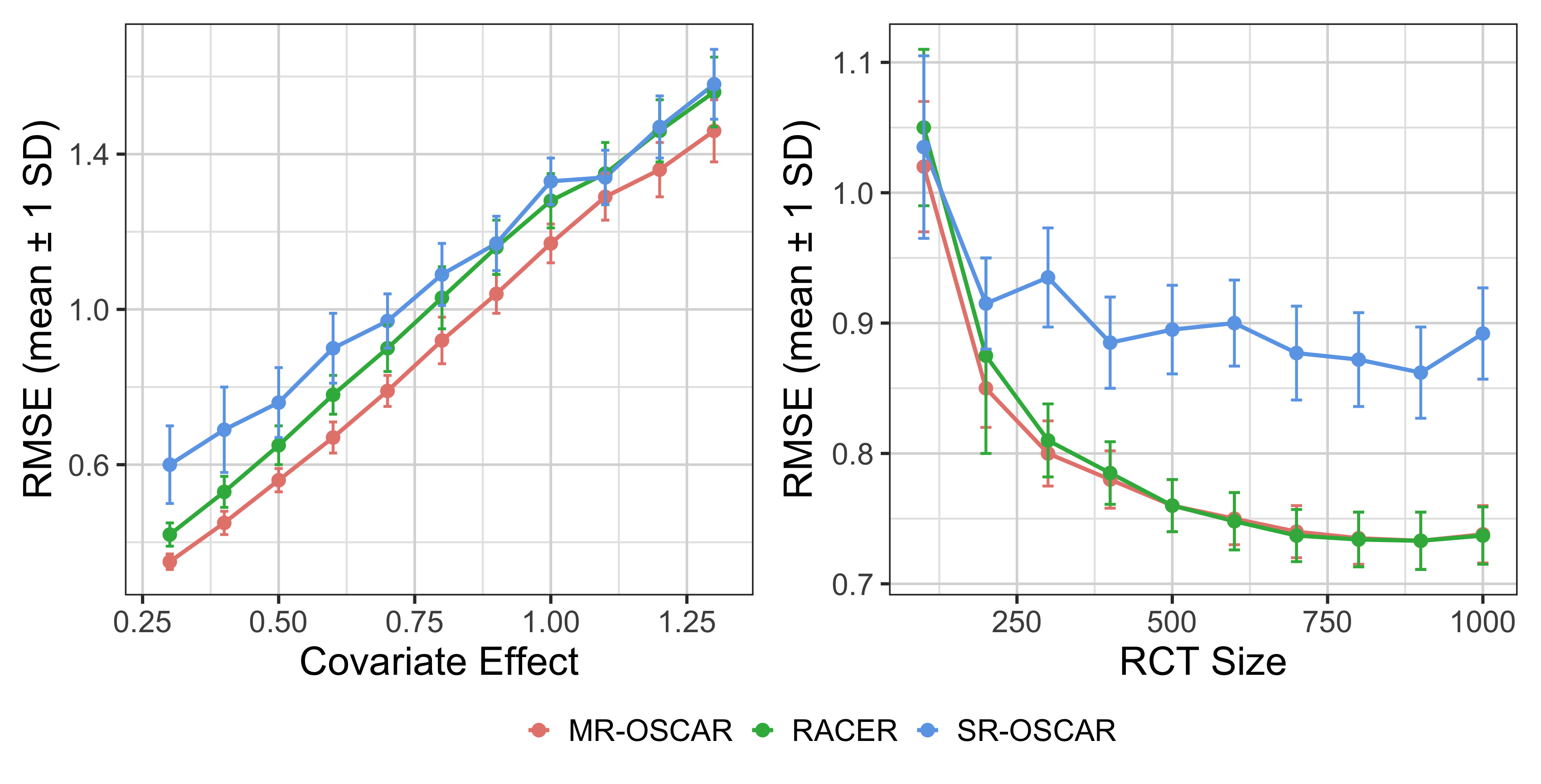}
    \caption{\textit{
\textbf{Left:} Mean RMSE (points) with $\pm 1$ SD (bars) versus the effect size of the OS-only covariates $\bV$ (OS size fixed at $n^o=1000$, RCT size at $n^r=300$).
\textbf{Right:} Mean RMSE versus RCT sample size $n^r$ (OS fixed at $n^o=1000$).}}
    \label{fig:Cov_eff}
\end{figure}

\subsection{Varying the Proportion of Shared Features}
We set $p_u=\lfloor f_1 p\rfloor$, $p_v=\lfloor f_2 p\rfloor$ for $f_1,f_2\in\{0,0.1,\ldots,0.5\}$ with $f_1+f_2\le 0.8$, and display heatmaps of $\text{RMSE}(\text{MR-OSCAR})-\text{RMSE}(\text{RACER})$ and $\text{RMSE}(\text{MR-OSCAR})-\text{RMSE}(\text{SR-OSCAR})$ in Figure~\ref{fig:Prop_shared}.

MR-OSCAR outperforms RACER (left panel, blue) when the OS contains many unavailable variables (large $f_2$) and the trial has few exclusive ones (small $f_1$). As $f_1$ increases, imputation cannot recover $\bU$-specific structure and RACER becomes preferable (red). MR-OSCAR uniformly dominates SR-OSCAR (right panel), confirming that augmenting beyond $\bZ$ through imputation and calibration is always preferable to restricting to the shared intersection alone.

\begin{figure}[htbp]
    \centering
    \includegraphics[width=0.49\linewidth]{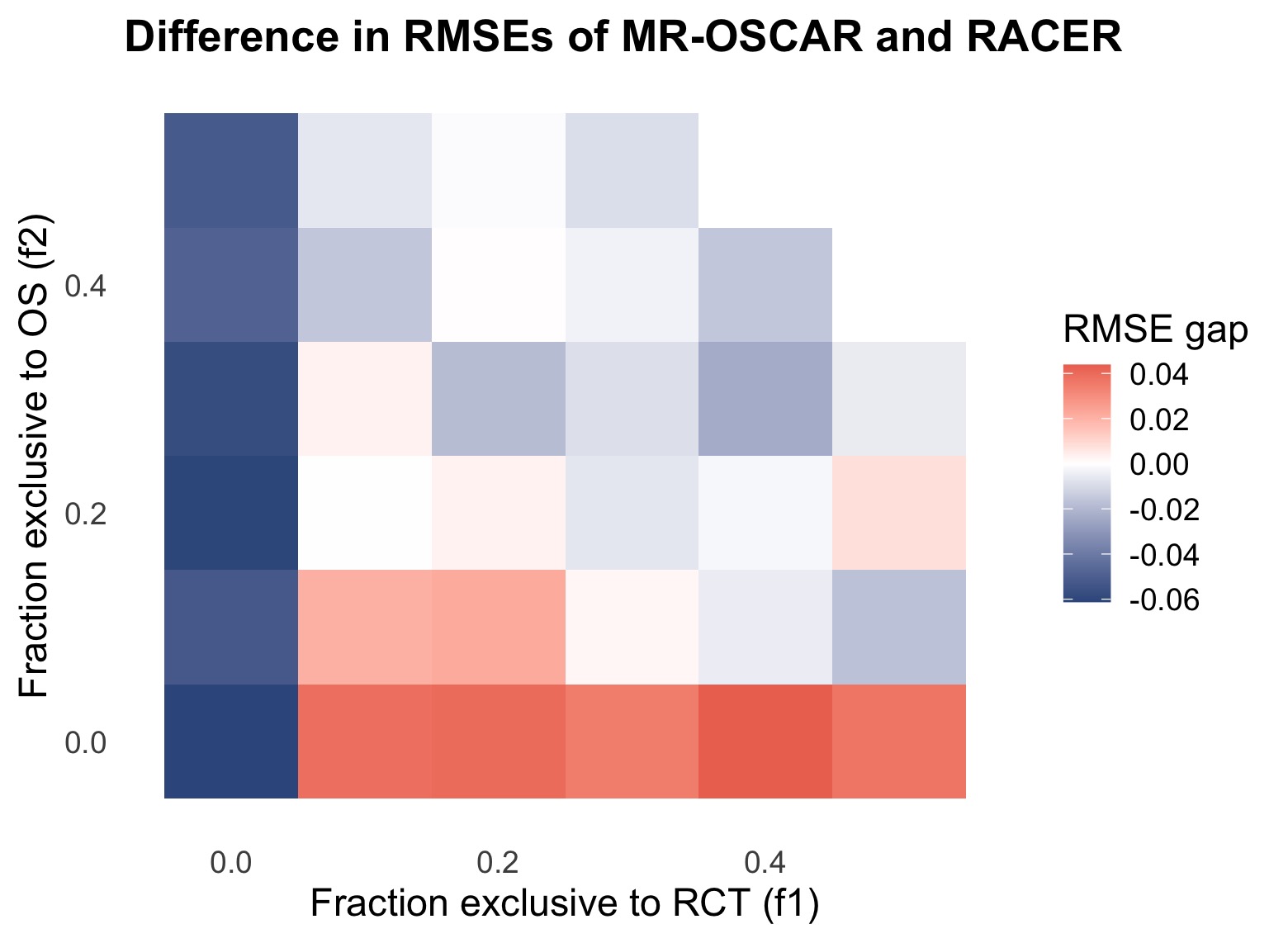}
    \includegraphics[width=0.49\linewidth]{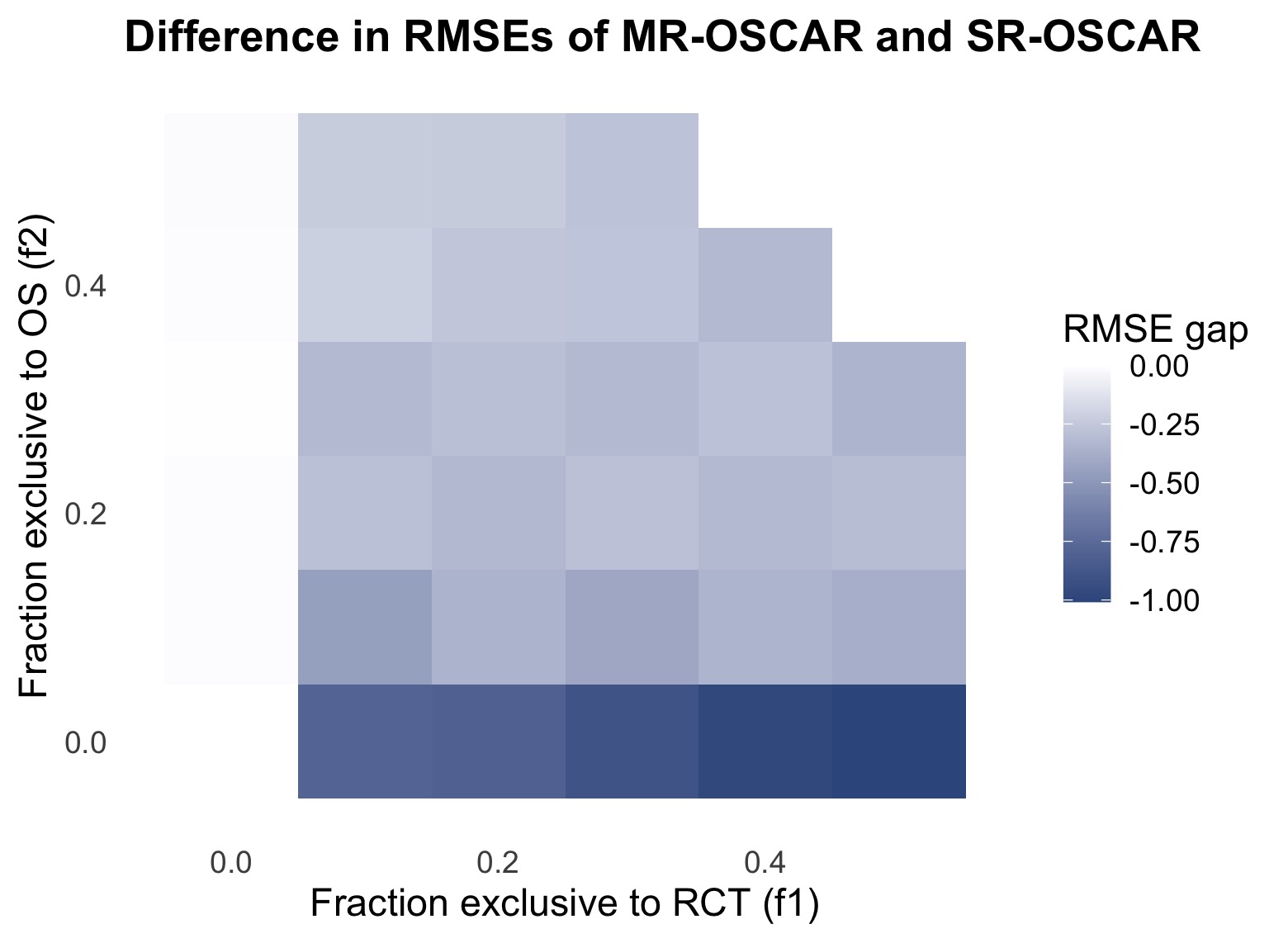}
    \caption{\textit{Heatmaps of RMSE gaps over the grid of covariate-mismatch fractions $(f_1,f_2)$, where $f_1$ is the share exclusive to the RCT ($\bU$) and $f_2$ the share exclusive to the OS ($\bV$). 
}}
    \label{fig:Prop_shared}
\end{figure}

Additionally, to assess how the benefit of imputation-augmented borrowing depends on the predictability of the
OS-only covariates, we carried out a sensitivity study that varies the population
$R^2(\bV\mid\bZ)$ from 0.1 to 0.9 (see Supplement, Figure~S\ref{fig:R2}).  
As $R^2(\bV\mid\bZ)$ increases, RMSE decreases for all methods, with MR-OSCAR
slightly worse than RACER when $\bV$ is poorly predictable but overtaking RACER once
$R^2(\bV\mid\bZ)\gtrsim 0.6$, while the SR-OSCAR benchmark is uniformly worst.

\section{Analysis of Greenlight Plus Study}\label{sec:real_data}

\subsection{Data Description and Pre-processing}
The trial sample comprises $n^r=860$ complete-case children from the Greenlight Plus Study \citep{heerman2022greenlight, heerman2024digital}, a multi-site RCT that evaluated whether augmenting standard well-child counseling with a literacy-sensitive digital program---including tailored text messages, goal-setting supports, and a caregiver dashboard---could help parents maintain healthy weight trajectories in early childhood.
Children were randomized approximately 1:1 to clinic-only care ($n=431$) or clinic plus digital intervention ($n=429$) across six sites (Duke, UNC, Vanderbilt, NYU, Stanford, and Miami).
Baseline covariates recorded in the RCT include child sex, birth weight, caregiver race/ethnicity, education, household income, preferred language (English or Spanish), food security, health literacy, and baseline weight-for-length $z$-score (WFLz).
The observational sample is drawn from electronic health records at three of the six trial sites with accessible systems---Duke ($n=27{,}679$), Vanderbilt ($n=10{,}498$), and UNC ($n=1{,}147$)---yielding $n^o=39{,}324$ non-enrolled children who received routine well-child care during the trial window and thus represent the control condition only.
EHR covariates include site, sex, race/ethnicity, language, insurance type, and baseline WFLz.
In the notation of \Cref{subsec:notation}, the shared block $\bZ$ comprises site, sex, race/ethnicity, language, and baseline WFLz; the RCT-only block $\bU$ includes birth weight, education, income, food security, and health literacy; and insurance, which is deliberately masked in the RCT to construct a covariate-mismatch scenario, serves as the OS-specific variable $\bV$.

For each child in both sources, WFLz outcomes at target ages were obtained by local linear interpolation of observed growth measurements within age-specific windows ($\pm1$ month at the 12- and 24-month endpoints), using the WHO Child Growth Standards.
The EHR cohort was assembled by linking patient demographic files with clinical encounter records containing anthropometric measurements from well-child and other visits.
Patients were excluded from the EHR analytic sample if they lacked a reliable weight measurement in the first 21 days of life at or above the third WHO percentile (71\% of exclusions), spoke a language other than English or Spanish (24\%), or had no coded well-child visit in their first year (4\%).
Insurance status for RCT participants at the three linked sites was extracted from the corresponding EHR billing records, and all variable names were harmonized across sources to ensure consistency before analysis.
The raw insurance field in the EHR contains several categories, including a large ``Unavailable/No Payer'' group comprising roughly 58\% of EHR patients. During binarization, this group was combined with the uninsured category ($V=0$), yielding a binary indicator $V\in\{0,1\}$ (insured vs.\ uninsured/unavailable). This coding inflates the proportion of $V=0$ observations and introduces site-specific sparsity in the insured category, which contributes to the variance inflation seen when the raw indicator is used directly in R-OSCAR (see \Cref{subsec:results_gp}).

\subsection{Method Implementations}
In our Greenlight Plus analysis we deliberately impose a covariate-mismatch
scenario by masking the binary insurance variable (that we call $V$) in the RCT and treating
it as unavailable for primary CATE estimation so that we can assess how well methods can perform with real world data relative to the gold standard when all covariates are available in both data sources.  The EHR data contain $V$ only
for control patients (there are no treated children in the external EHR), so
any borrowing must respect this restriction and cannot rely on external
information about the treatment arm. In this scenario, SR-OSCAR implements the R-OSCAR
framework using only the shared covariates $Z$: an outcome model for the
control arm is first trained in the EHR on $(Z,Y)$ for $A=0$ and transported to
the RCT, a discrepancy function is learned on RCT controls to correct for
outcome shift, and finally a treatment-arm model and CATE calibration are
fit using RCT data. For MR-OSCAR, in the first step, we use the EHR {control} sample to learn a
prediction model for insurance based on the shared covariates, fitting a
logistic LASSO for $\Pr(V=1 \mid Z)$
using only individuals with $A=0$ in the EHR.  We then apply this fitted model
to the RCT controls and treated children to obtain subject-specific
probabilities $\widehat V = \Pr(V=1 \mid Z)$. Next, we fit the EHR control outcome model on $(Z,V)$, i.e., using both the
shared covariates and the {observed} insurance in the large EHR control
sample.  In the RCT, by contrast, all outcome, discrepancy, and CATE
calibration models use $(Z,\widehat V)$, replacing the unobserved $V$ with its
logistic-LASSO prediction.  In this way, MR-OSCAR borrows external information
only from the control arm and only through the shared covariates $Z$ and the
imputed insurance $\widehat V$, while the randomized treatment contrast in the
trial remains the sole source of causal identification.
We compare against three baselines.
SR-OSCAR is described above (shared covariates only).
R-OSCAR \citep{asiaee2023leveraging} uses $(Z,V)$ in both the EHR and the RCT, representing an ``oracle'' scenario in which insurance is fully available in the trial (it is masked only to showcase MR-OSCAR's utility).
RACER \citep{asiaee2023leveraging} is an RCT-only benchmark that fits separate outcome models for $A=0$ and $A=1$ using all trial covariates and calibrates the CATE via pseudo-outcome regression.

\subsection{Results}\label{subsec:results_gp}

Figure~\ref{fig:sorted_cate} (upper panel) displays the sorted individual CATE estimates and their 95\% confidence intervals for all four methods.  All approaches reveal
substantial heterogeneity in the intervention effect across children, with estimated CATEs ranging from moderately negative to clearly positive.  The
width of the bands, however, differs markedly: MR-OSCAR produces the tightest intervals, followed by SR-OSCAR, then R-OSCAR, with RACER showing the widest intervals.  These visual impressions are confirmed by
Table~\ref{tab:cate_widths}.  Relative to the RCT-only RACER benchmark, MR-OSCAR reduces the mean CI width from 1.03 to 0.83 (roughly a 20\% reduction) and the mean radius from 0.283 to 0.178, with SR-OSCAR and
R-OSCAR yielding intermediate gains. An interesting feature of Table~\ref{tab:cate_widths} is that R-OSCAR, which has access to the {true} insurance indicator $V$ in the RCT, yields slightly wider CATE intervals than MR-OSCAR, which uses an imputed
$\widehat V$. This does not necessarily contradict the intuition that more covariate information should help; rather, it may reflect how that information enters the
estimation pipeline. In the Greenlight Plus data, our covariate-balance diagnostics
(\Cref{fig:loveplot}, Supplement) show that the collapsed
insurance indicator is not only a little imbalanced across treatment arms, but it is also relatively rare and exhibits site-specific sparsity in the uninsured category. Including the raw insurance factor $V$ directly in the regression steps therefore forces the RCT models to estimate several small cell means and interactions, which can substantially inflate variance with little corresponding reduction in bias.
By contrast, MR-OSCAR first learns a smooth prediction
$\widehat V = \Pr(V=1 \mid Z)$ in the large EHR control sample and then uses this denoised surrogate in the RCT outcome and calibration models, stabilizing the contribution of insurance. As noted in the data description, the ``Unavailable/No Payer'' category was merged with uninsured during binarization, inflating the $V\!=\!0$ counts and exacerbating the sparsity that drives R-OSCAR's variance.  MR-OSCAR sidesteps this issue by replacing the noisy binary indicator with a smooth imputed probability, effectively regularizing the insurance signal. As a result, MR-OSCAR can
achieve shorter intervals than the oracle R-OSCAR in this particular covariate-mismatch setting, while still relying solely on the randomized
trial for causal identification.

\begin{table}[htbp]
\centering
\resizebox{0.85\linewidth}{!}{%
\begin{tabular}{lccccc}
\hline
Method   & Mean width & Median width & 25th pct.\ width & 75th pct.\ width & Mean radius \\
\hline
MR-OSCAR & \textbf{0.831} & \textbf{0.786} & \textbf{0.716} & \textbf{0.901} & \textbf{0.178} \\
SR-OSCAR & 0.872 & 0.830 & 0.764 & 0.936 & 0.195 \\
R-OSCAR  & 0.909 & 0.846 & 0.770 & 0.962 & 0.217 \\
RACER    & 1.030 & 0.970 & 0.832 & 1.190 & 0.283 \\
\hline
\end{tabular}}
\caption{Summary of individual CATE uncertainty across methods. 
For each method, we report the mean, median, 25th and 75th percentiles of the 95\% CI widths, and the mean
radius $r_i = 0.25\cdot(\text{width}_i)^2$ averaged over trial participants.}
\label{tab:cate_widths}
\end{table}

\begin{figure}[htbp]
    \centering
    \includegraphics[width=0.75\linewidth]{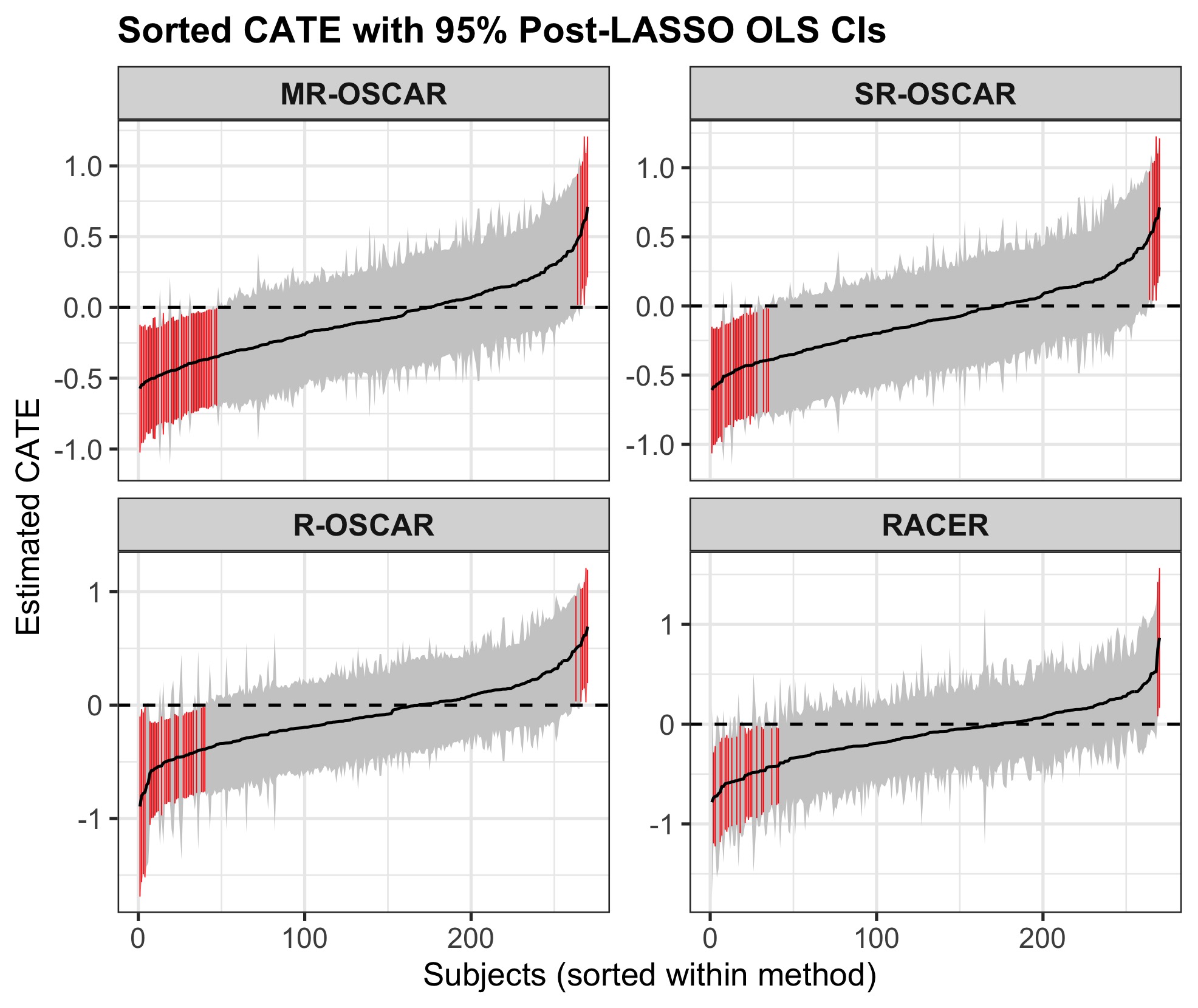}
    \includegraphics[width=0.8\linewidth]{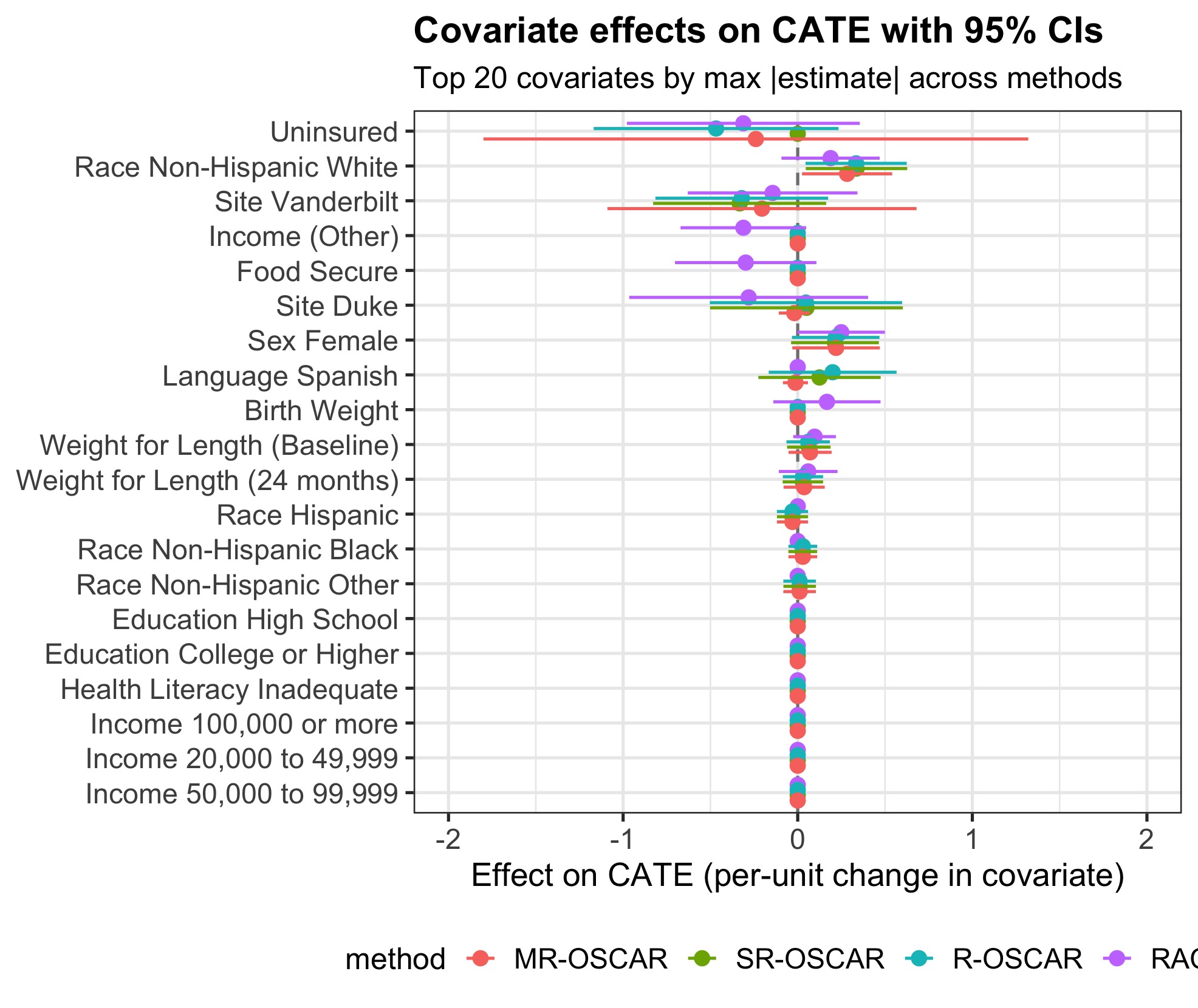}
    \caption{\small\textit{Upper Panel:} Sorted individual CATE estimates with 95\% post-LASSO OLS confidence intervals; grey intervals cover 0, while red intervals correspond to subjects whose CIs exclude 0. \textit{Lower Panel:} Covariate-specific effects on the CATE with 95\% confidence intervals for the top 20 covariates (by maximum absolute effect size).}
    \label{fig:sorted_cate}
\end{figure}


Figure~\ref{fig:sorted_cate} (lower panel) examines how key covariates affect the CATE.
Across methods, being uninsured is associated with lower predicted treatment
benefit, whereas non-Hispanic White race, Vanderbilt site, and higher-income
categories are associated with larger CATEs.  Baseline anthropometrics (birth
weight, baseline weight-for-length, and 24-month weight-for-length),
language, and health literacy also exhibit non-negligible associations.  The
four estimators largely agree on the direction and relative magnitude of
these effects; the borrowing-based methods (MR-OSCAR and SR-OSCAR) tend to
produce slightly shrunk and more stable coefficient estimates than RACER,
consistent with their narrower CATE intervals.

\section{Discussion}\label{sec:conclusion}
We have proposed MR-OSCAR, a framework for improving CATE estimation in randomized trials by borrowing strength from observational sources that record additional, potentially effect-modifying covariates not available in the trial.
Rather than discarding the mismatched covariates, MR-OSCAR imputes them in the trial using a prediction model trained in the observational sample, then calibrates all outcome and CATE models back to the RCT, preserving the randomized treatment contrast as the sole source of causal identification.

Our error bounds decompose the CATE risk into the complexity of the final CATE regression, the quality of OS-to-RCT outcome model transport, and the accuracy of covariate imputation.
When OS-only covariates are strongly related to the outcome and moderately predictable from shared variables, MR-OSCAR strictly dominates both an RCT-only estimator and a shared-only borrowing strategy.
When the missing block is weakly predictive or hard to impute, MR-OSCAR gracefully shrinks back toward RCT-only estimation without substantial additional error.
These theoretical predictions are confirmed by simulations across a range of covariate-mismatch configurations.

In the Greenlight Plus application, MR-OSCAR delivers the narrowest individual confidence intervals---roughly a 20\% reduction in mean CI width relative to the RCT-only benchmark---even outperforming an ``oracle'' estimator that observes the true insurance variable in the trial.
This counterintuitive result arises because a smooth imputed probability regularizes a sparse, noisy binary covariate more effectively than using the raw indicator directly.
The covariate-effect analysis further shows that insurance status is a key driver of treatment response, underscoring the cost of ignoring mismatched variables and the value of imputation-based integration for identifying which children benefit most from obesity prevention programs.

\if1\blind
{
\section{Acknowledgement}
%
This research was funded by the Patient-Centered Outcomes Research Institute (PCORI) award ME-2023C1-32148, ``Improving Methods for Conducting Patient-Centered Comparative Clinical Effectiveness Research.'' All authors are supported by this award.
}\fi 

\if1\blind
{
\section{Data Availability Statement}
%
This authors are unable to share the data due to patient privacy concerns. However, users may request access from Vanderbilt University Medical Center upon reasonable request. 
}\fi 

\bibliographystyle{apalike}

\bibliography{ref}
\newpage
\begin{center}
{\large\bf SUPPLEMENTARY MATERIAL}
\end{center}

\begin{table}[!ht]
\centering
\caption{Comprehensive glossary of notation.}\label{tab:notation}
{\scriptsize
\begin{tabular}{@{}c l p{0.72\textwidth}@{}}
\toprule
& \textbf{Symbol} & \textbf{Definition / description} \\
\midrule
\multirow{4}{*}{\rotatebox{90}{\scriptsize\textit{Data}}}
& $S\in\{r,o\}$ & Source indicator (RCT or OS) \\
& $Y;\; Y(1),Y(-1)$ & Observed / potential outcomes \\
& $A\in\{-1,1\}$ & Treatment indicator \\
& $n^s;\; n_a^s$ & Total and arm-specific sample sizes for source $s$ \\[2pt]
\midrule
\multirow{5}{*}{\rotatebox{90}{\scriptsize\textit{Covariates}}}
& $\bX=(\bU,\bZ,\bV)\in\R^{p}$ & Complete covariate vector;\; $p=p_u\!+\!p_z\!+\!p_v$ \\
& $\bX^r=(\bU,\bZ)\in\R^{p_r}$ & RCT-observed covariates;\; $p_r=p_u\!+\!p_z$ \\
& $\bX^o=(\bZ,\bV)\in\R^{p_o}$ & OS-observed covariates;\; $p_o=p_z\!+\!p_v$ \\
& $\bU,\;\bZ,\;\bV$ & RCT-exclusive, shared, OS-exclusive blocks \\
& $\widehat\bV=\hat g(\bZ)$ & Imputed OS-only covariates in the RCT \\[2pt]
\midrule
\multirow{8}{*}{\rotatebox{90}{\scriptsize\textit{Core quantities}}}
& $\mu^s_a(\bm x)$ & $\E[Y\mid \bX\!=\!\bm x,A\!=\!a,S\!=\!s]$: conditional mean outcome, arm $a$, source $s$ \\
& $\delta_a(\bm x)$ & $\mu^r_a(\bm x)-\mu^o_a(\bm x)$: outcome shift between RCT and OS \\
& $\tau^r(\bm x^r)$ & $\E[Y(1)\!-\!Y(-1)\mid \bX^r\!=\!\bm x^r,S\!=\!r]$: target CATE \eqref{eq:rct-cate} \\
& $\pi^r_a(\bm x^r)$ & $P(A\!=\!a\mid \bX^r\!=\!\bm x^r,S\!=\!r)$: RCT propensity \eqref{eq:rct-propensity} \\
& $\mu^r(\bm x^r)$ & $\sum_a \pi^r_{-a}\mu^r_a$: counterfactual mean outcome (CMO) \eqref{eq:cmo} \\
& $\tilde\mu^r(\bm x^r)$ & $\E[\mu^r(\bX)\mid \bX^r\!=\!\bm x^r,S\!=\!r]$: marginalized CMO \eqref{eq:rct-cmo} \\
& $\tau_m(\bX^r,A,Y)$ & $A\{Y-m(\bX^r)\}/\pi^r_A(\bX^r)$: pseudo-outcome \eqref{eq:pseudo-outcome} \\
& $\hat\tau_{\text{SR}},\;\hat\tau_{\text{MR}}$ & SR-OSCAR and MR-OSCAR CATE estimators (see Table~\ref{tab:cmo-summary}) \\[2pt]
\midrule
\multirow{6}{*}{\rotatebox{90}{\scriptsize\textit{Estimators}}}
& $\htmu^r,\;\htmu^{\mathrm{sh}},\;\htmu^{\mathrm{im}}$ & CMO estimators: RACER, SR-OSCAR, MR-OSCAR (Table~\ref{tab:cmo-summary}) \\
& $\hat\mu^{o,\mathrm{sh}}_a,\;\hat\mu^{o,\mathrm{im}}_a$ & OS outcome models (shared-only $\bZ$; full $\bX^o$) \\
& $\hat\delta^{\mathrm{sh}}_a,\;\hat\delta^{\mathrm{im}}_a$ & Discrepancy: \eqref{eq:z-only-discrepancy}, \eqref{eq:imputation-discrepancy} \\
& $\hat\mu^{\mathrm{sh}}_a,\;\hat\mu^{\mathrm{im}}_a$ & Calibrated arm means: $\hat\mu^{o,(\cdot)}_a+\hat\delta^{(\cdot)}_a$ \\
& $\hat g$ & Imputation map $\bZ\to\bV$ \\
& $\hat\mu^r_a$ & RACER arm-specific RCT outcome model \\[2pt]
\midrule
\multirow{4}{*}{\rotatebox{90}{\scriptsize\textit{Classes}}}
& $\cM^{o,\mathrm{sh}}_a,\;\cM^{o,\mathrm{im}}_a$ & OS outcome classes on $\bZ$ ($p_z$-dim) and $\bX^o$ ($p_o$-dim) \\
& $\cD^{\mathrm{sh}}_a,\;\cD^{\mathrm{im}}_a;\;\cD$ & Discrepancy on $\bX^r$ ($p_r$-dim) / $(\bX^r,\widehat\bV)$ ($p$-dim);\; CATE class on $\bX^r$ \\
& $\cM^r_a;\;\cG$ & RACER outcome class on $\bX^r$;\; imputation class on $\bZ$ \\
& $\mathcal{P}^{(\cdot)}_a(\cdot)$ & Penalty term matching the corresponding class \\[2pt]
\midrule
\multirow{5}{*}{\rotatebox{90}{\scriptsize\textit{Risk terms}}}
& $\Delta_2^2(\cH,\tau^r)$ & $\inf_{h\in\cH}\E\{(h(\bX^r)-\tau^r(\bX^r))^2\mid S\!=\!r\}$: approximation error \\
& $\Rad_n(\cH)$ & $\E_{\bm\epsilon}\bigl[\sup_{h\in\cH}\frac{1}{n}\sum_{i=1}^{n}\epsilon_i\,h(\bX_i^r)\bigr]$: empirical Rademacher complexity ($\epsilon_i\!\stackrel{\mathrm{iid}}{\sim}\!\mathrm{Unif}\{-1,1\}$) \\
& $c(\cH)$ & Localized complexity parameter \citep{bartlett2006local} \\
& $B_{\mathrm{sh}}^2$ & $\sum_a\inf_{h\in\cH^{\mathrm{sh}}_a}\E[(h(\bX^r)-\tilde\mu^r_a(\bX^r))^2\mid S\!=\!r]$: mismatch penalty \eqref{eq:BZ-def} \\
& $r_{\mathrm{im}}^2$ & $\E[\|\bV-g^\star(\bZ)\|^2\mid S\!=\!r]$: imputation risk \\
& $L$ & Lipschitz constant of $\mu^o_a$ in $\bV$ \\[2pt]
\midrule
\multirow{5}{*}{\rotatebox{90}{\scriptsize\textit{Sparsity}}}
& $s_\tau$ & Sparsity of the CATE model on $\bX^r$ (\S\ref{subsec:linear-example}) \\
& $s_\mu^{\mathrm{sh}},\;s_\mu^{\mathrm{im}}$ & OS outcome sparsity on $\bZ$ (SR) and $\bX^o$ (MR);\; $s_\mu^{\mathrm{sh}}\le s_\mu^{\mathrm{im}}$ \\
& $s_\delta^{\mathrm{sh}},\;s_\delta^{\mathrm{im}}$ & Discrepancy sparsity on $\bX^r$ (SR) and $(\bX^r,\widehat\bV)$ (MR);\; $s_\delta^{\mathrm{sh}}\le s_\delta^{\mathrm{im}}$ \\
& $s_j$ & Per-coordinate imputation sparsity: $s_j:=\|\lambda_{o,j}\|_0$ \\
\bottomrule
\end{tabular}
}
\end{table}
\clearpage
Algorithm~\ref{alg:high-level} summarizes the common structure of these procedures.
\begin{algorithm}[th]
\footnotesize
\caption{High-level structure of R-OSCAR-type estimators under covariate mismatch}
\label{alg:high-level}
\begin{algorithmic}[1]
\State \textbf{Input:} 
RCT data $\{(\bX_i^r,A_i^r,Y_i^r)\}_{i=1}^{n^r}$, 
OS data $\{(\bX_i^o,A_i^o,Y_i^o)\}_{i=1}^{n^o}$, 
choice of estimator in $\{\text{RACER, SR-OSCAR, MR-OSCAR}\}$.
\State \textbf{If} using MR-OSCAR \textbf{then} 
  estimate $g:\bZ\mapsto\bV$ in the OS and impute 
  $\widehat\bV_i = \hat g(\bZ_i^r)$ for RCT units.
\For{$a \in \{-1,1\}$}
  \State \textbf{OS outcome modeling:}
    \begin{itemize}
      \item RACER: no OS outcome model is used.
      \item SR-OSCAR: fit an OS model $\hat\mu^{o,\mathrm{sh}}_{a}(\bZ)$ on $(\bZ^o,Y^o)$.
      \item MR-OSCAR: fit an OS model $\hat\mu^{o,\mathrm{im}}_{a}(\bX^o)$ 
            on $(\bZ^o,\bV^o,Y^o)$.
    \end{itemize}
  \State \textbf{Calibration to the RCT:}
    \begin{itemize}
      \item RACER: fit $\hat\mu^r_a(\bX^r)$ directly on the RCT.
      \item SR-OSCAR: fit a discrepancy $\hat\delta^{\mathrm{sh}}_{a}(\bX^r)$ so that 
            $Y^r \approx \hat\mu^{o,\mathrm{sh}}_{a}(\bZ^r) + \hat\delta^{\mathrm{sh}}_{a}(\bX^r)$.
      \item MR-OSCAR: fit a discrepancy $\hat\delta^{\mathrm{im}}_{a}(\bX^r,\widehat\bV)$ so that 
            $Y^r \approx \hat\mu^{o,\mathrm{im}}_{a}(\bZ^r,\widehat\bV) 
                    + \hat\delta^{\mathrm{im}}_{a}(\bX^r,\widehat\bV)$.
    \end{itemize}
\EndFor
\State \textbf{Construct arm-specific calibrated predictions} 
      $\hat\mu^{\mathrm{sh}}_a$ (SR-OSCAR) or $\hat\mu^{\mathrm{im}}_a$ (MR-OSCAR) and the associated CMO for each individual by 
      combining the two arm-specific calibrated predictions with 
      swapped RCT randomization weights (as in \cite{asiaee2023leveraging}).
\State \textbf{Form transformed outcomes} by subtracting the CMO from the observed 
      outcome in the RCT and multiplying by a known function of \((A,\bX^r)\) 
      (the pseudo-outcome), and regress these transformed outcomes on 
      \(\bX^r\) to obtain the final CATE estimator \(\hat\tau(\bX^r)\).
\State \textbf{Output:} 
  CATE estimator $\hat\tau(\cdot)$ corresponding to RACER, SR-OSCAR, or MR-OSCAR.
\end{algorithmic}
\end{algorithm}
\subsection*{Simulation Setup Details}
We generate a latent baseline vector $\bX=(\bU,\bZ,\bV)\in\mathbb{R}^{p}$ from an equicorrelated Gaussian distribution with correlation $\rho=0.4$ and unit marginal variances. To induce a mild covariate (domain) shift between sources, we set $\E[\bX\mid S=o]=\bm{0}$ and $\E[\bX\mid S=r]=(\bm{\delta}_{u},\bm{0}_{z},\bm{0}_{v})$, where the coordinates of $\bm{\delta}_{u}\in\mathbb{R}^{p_u}$ are drawn independently with random sign and magnitude in a small range (e.g., $|\delta_{u,j}|\in[0.25,0.5]$). Treatment is randomized in the RCT with $\Pr(A=1)=0.5$, whereas in the OS it follows a sparse logistic propensity model depending only on a subset of shared covariates, $\Pr(A=1\mid \bZ)=\mathrm{expit}(\alpha_0+\bZ_S^\top\bm \gamma)$, where $S\subset\{1,\dots,p_z\}$ is a fixed index set (e.g., $|S|=10$), $\bm \gamma$ has nonzero entries on $S$ (e.g., $\gamma_j\in[0.25,0.5]$ with random sign), and $\alpha_0$ is chosen so that the marginal treated proportion is approximately $1/3$. Potential outcomes are linear with sparse coefficients: for $a\in\{0,1\}$ we set $Y(a)=\bX^\top\bm \beta_a+\varepsilon_a$, where a small proportion of coordinates in $\bX^r=(\bU,\bZ)$ carry signal (e.g., $5\%$ with magnitude about $2/3$) and coefficients on $\bV$ have larger magnitude (about $1$), reflecting outcome-relevant information observed only in the OS. To induce an outcome-model shift, we perturb a small fraction (e.g., $2\%$) of the $\bX^r$ coefficients in the RCT relative to the OS. Finally, we add homoskedastic Gaussian noise independently across units (and arms), $\varepsilon_a\sim \mathcal{N}(0,\sigma_\varepsilon^2)$; in our implementation the noise standard deviation is $0.5\,\sigma$, so $\sigma_\varepsilon^2=\sigma^2/4$.
\subsection*{Additional Simulation: Varying the predictability of OS-only covariates}\label{sec:sim-R2}
Here, in order to study the effect of imputability of $\bV$ from $\bZ$, we vary the population predictability $R^2(\bV\mid \bZ)$
by sampling $\bV\mid \bZ=\bm\Lambda \bZ+\varepsilon$ with $\varepsilon\sim\mathcal N(0,\sigma^2_{\bV\mid \bZ}\bm{I})$ and
choosing $\sigma^2_{\bV\mid \bZ}$ to achieve target values $R^2\in\{0.10,0.20,\ldots,0.90\}$, while keeping
$\bm\Lambda$ fixed and the rest of the study design same as before.
\begin{figure}[htbp]
    \centering
    \includegraphics[width=0.6\linewidth]{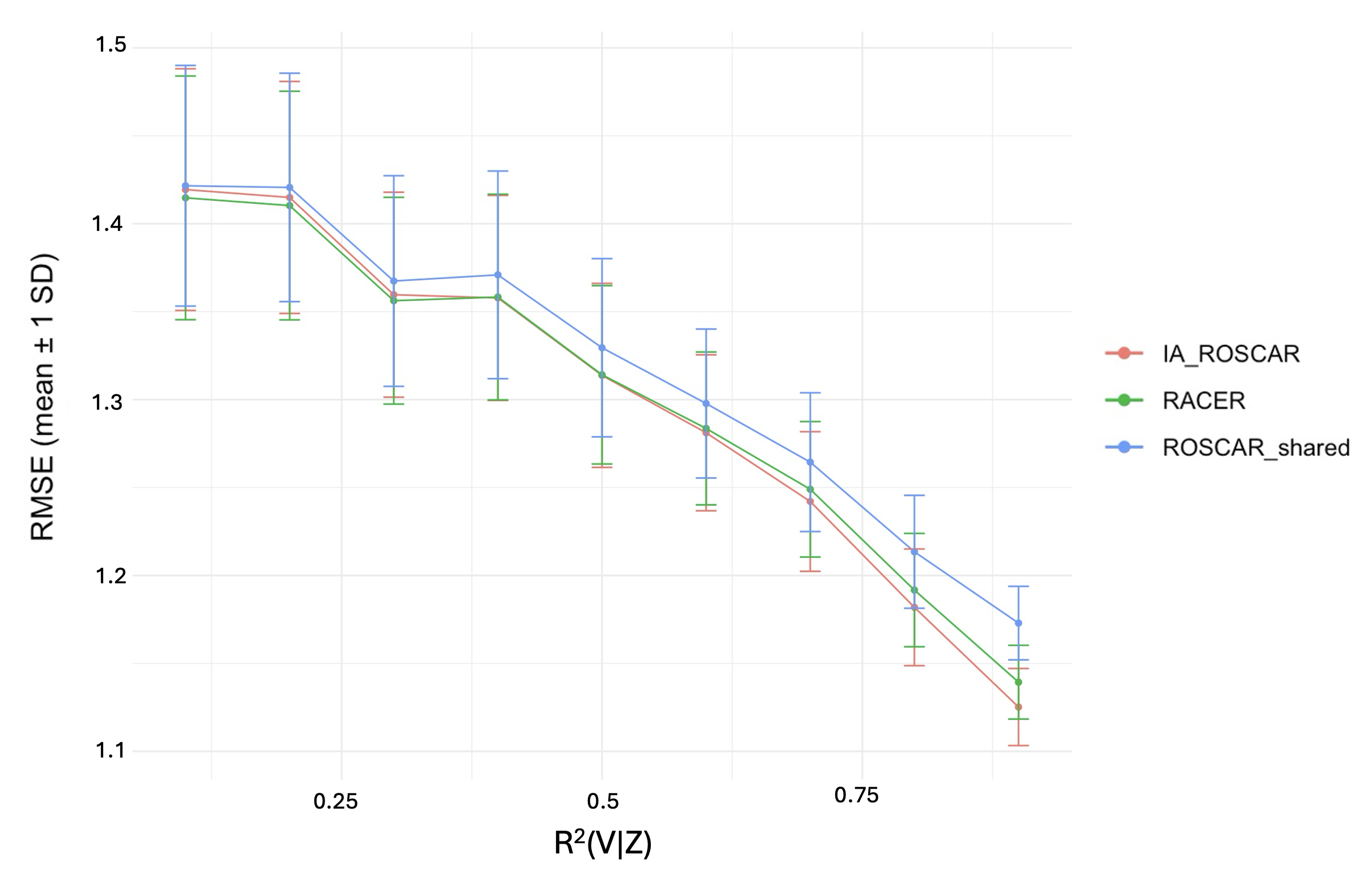}
    \caption{\textit{RMSE (mean $\pm$ 1 SD) for CATE estimation in the RCT population as a function of
  $R^2(\bV\mid \bZ)$, which controls how well the OS-only covariates $\bV$ can be predicted from $\bZ$.
  MR-OSCAR (red) is slightly worse than RACER (green) when imputability is low, ties around
  moderate values, and dominates once $R^2(\bV\mid \bZ)$ is sufficiently large; SR-OSCAR (blue)
  is uniformly worst.}}
    \label{fig:R2}
\end{figure}
Figure~\ref{fig:R2} displays RMSE versus $R^2(\bV\mid \bZ)$.
Three qualitative features emerge:
(i) {monotone improvement with imputability.}  All methods enjoy decreasing RMSE as $R^2(\bV\mid \bZ)$
increases, reflecting the fact that the problem becomes easier when $\bV$ is more predictable from $\bZ$;
(ii) {threshold behavior for MR-OSCAR.}  At low levels of imputability ($R^2\approx 0.10$-$0.30$),
MR-OSCAR is slightly {worse} than RACER: the noise introduced by imputing a poorly predictable
$\bV$ offsets the benefit of borrowing from the OS.  In the mid-range ($R^2\approx 0.40$-$0.60$), the two
methods are essentially indistinguishable.  Beyond a clear threshold ($R^2\gtrsim 0.60$),
MR-OSCAR overtakes RACER and the gap widens as $R^2$ increases; this aligns with our theory in which
the MR-OSCAR risk improves as the imputation error $r_{\mathrm{im}}^2$ shrinks;
(iii) {dominance over $\bZ$-only borrowing.}  SR-OSCAR is uniformly worse across the entire
range, because it ignores the $\bV$ signal entirely; it never crosses the other curves.
Error bars shrink modestly with $R^2$, indicating greater stability when $\widehat{\bV}$ is reliable.
\subsection*{Covariate Balance in the Greenlight Plus RCT}
At first glance it may seem surprising that MR-OSCAR, which treats
insurance as unobserved in the RCT and instead uses an imputed
$\widehat V$, can outperform an ``oracle'' R-OSCAR that conditions on
the true insurance indicator $V$.
The covariate-balance diagnostics help clarify this phenomenon. To assess whether any gains from our borrowing strategies could be
attributed to residual imbalances in baseline covariates, we examined
covariate balance between treatment arms in the Greenlight Plus RCT.
For each baseline covariate $W$, we computed the absolute
standardized mean difference (SMD) between $A=1$ and $A=0$,
defined as
\[
\operatorname{SMD}(W)
=
\frac{\bar W_{1} - \bar W_{0}}
     {\sqrt{\{s_{1}^{2} + s_{0}^{2}\}/2}},
\]
where $\bar W_{a}$ and $s_{a}^{2}$ denote the sample mean and variance
of $W$ in arm $A=a$; for binary covariates we used the analogous
standardization based on the pooled binomial variance.  Figure~\ref{fig:loveplot}
displays the resulting ``love plot'' of absolute SMDs.

\begin{figure}
    \centering
    \includegraphics[width=0.85\textwidth]{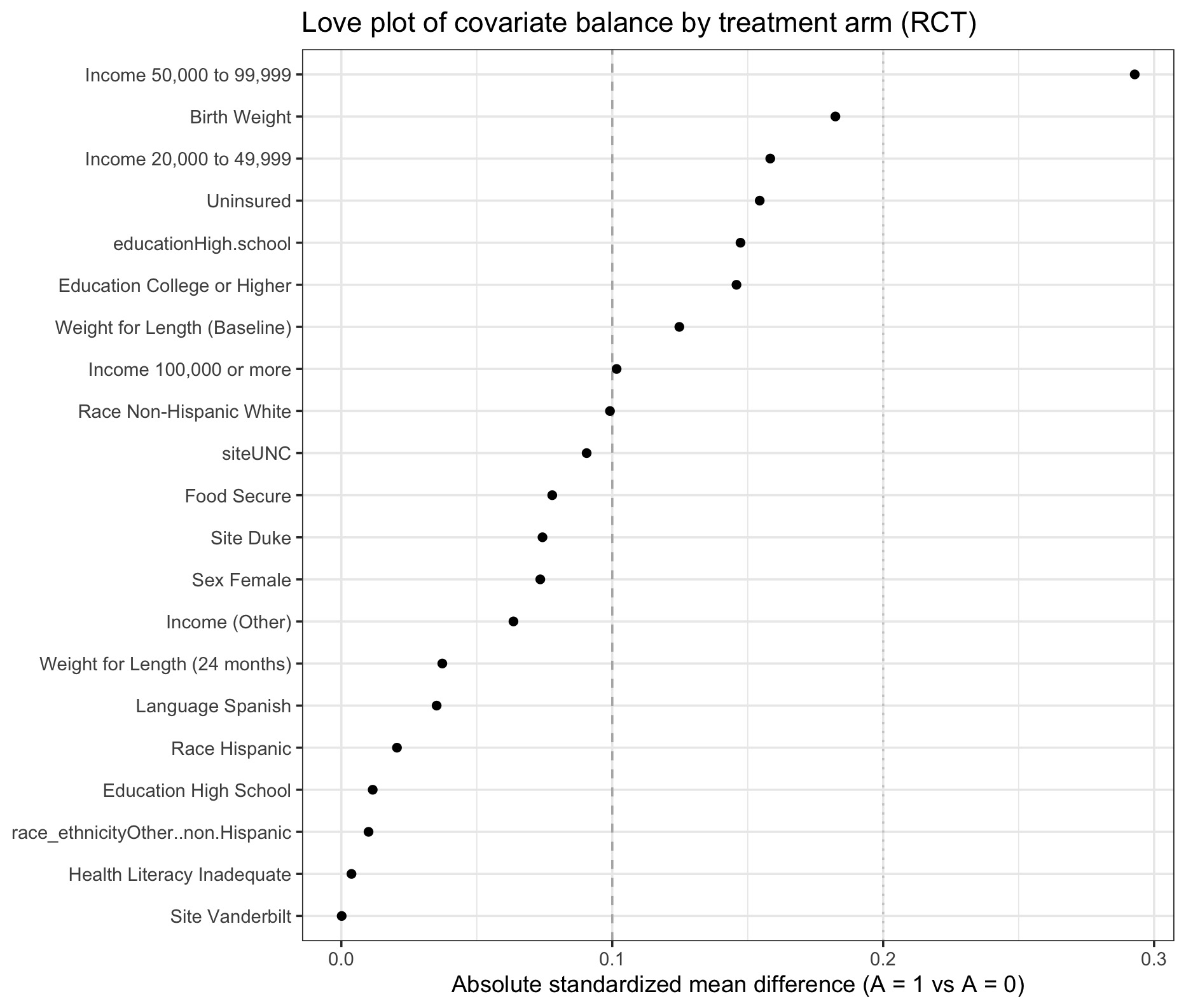}
    \caption{Loveplot of Covariate Balance}
    \label{fig:loveplot}
\end{figure}

Overall, covariate balance was excellent.  The vast majority of
baseline covariates had absolute SMDs well below $0.10$, a common
benchmark for negligible imbalance.  The only clearly pronounced
difference was for the income category ``\$50{,}000-\$99{,}999'',
which exhibited an absolute SMD of approximately $0.30$; most remaining
covariates were below $0.15$.  In particular, the collapsed insurance
indicator (Uninsured vs.\ Insured) showed a modest absolute SMD slightly above $0.15$. This indicates that
randomization achieved a moderate balance on insurance and, and that any efficiency gains from MR‐OSCAR
may be driven by correcting a sizable treatment‐arm imbalance in
insurance. Rather, as discussed below, the imputation step in MR‐OSCAR
acts primarily to stabilize a relatively sparse and site‐heterogeneous
insurance signal when incorporating EHR information into the CATE
estimation. Moreover, insurance is a relatively sparse and heterogeneous covariate:
only a small fraction of children are uninsured, and the uninsured
counts are unevenly distributed across sites (e.g., very few uninsured
children at Duke and UNC).
When R-OSCAR includes the full insurance factor $V$ (and its associated
basis functions) in the RCT outcome and calibration regressions, the
post-LASSO OLS steps must estimate coefficients for several small cells,
leading to unstable estimates and inflated sampling variance.

MR-OSCAR handles insurance differently.
It first fits an imputation model $g(Z) \approx \E[V \mid Z]$ in the
large EHR control sample and then propagates only the smoothed
prediction $\widehat V = g(Z)$ into the RCT outcome and calibration
models.
This construction pools information across sites and shrinks away
idiosyncratic noise and coding irregularities in the observed insurance
variable, effectively replacing a sparse binary factor by a stable,
continuous risk score.
The bias incurred by using $\widehat V$ instead of $V$ is negligible,
whereas the variance reduction from smoothing can be substantial.
This variance-bias tradeoff explains why, in the GPS application,
MR-OSCAR attains shorter CATE intervals than oracle R-OSCAR despite
not observing $V$ directly in the RCT.

\subsection*{Error bounds for SR-OSCAR on the shared covariate space}
\label{subsec:theory-sroscar}

We first consider the shared-only estimator SR-OSCAR.
Let \(\hat\tau_{\text{SR}}\) be the CATE estimator constructed in
Section~2.1.1, and suppose the last-stage calibrator
\(\hat\delta^{\mathrm{sh}}\) in \eqref{eq:shared-stage2-calib} is fit in a
class \(\cD\) with complexity measured by \(\Rad_{n^r}(\cD)\). Define the arm-specific marginalized RCT CMOs
\[
  \tilde\mu^r_a(\bm x^r)
    := \E\bigl[
          \mu^r_a(\bm x^r,\bV)
        \mid \bX^r=\bm x^r, S = r
       \bigr],
\]
and let \(m_{\bZ,a}(\bX^r)\) denote the augmentation component used by
SR-OSCAR for arm \(a\) (the calibrated OS prediction entering the CMO).
We measure the {augmentation error} by
\[
  \Delta_2^2(m_{\bZ,a},\tilde\mu^r_a\mid\bX^r)
    := \E\Bigl[
          \bigl\{
            m_{\bZ,a}(\bX^r)
            - \tilde\mu^r_a(\bX^r)
          \bigr\}^2 \mid S = r
        \Bigr].
\]

\begin{theorem}[Risk bound for SR-OSCAR]
\label{thm:sroscar}
Suppose Assumptions~\ref{as:ident} and \ref{as:shift}-\ref{as:rates}
hold, and the nuisance estimators are obtained with cross-fitting across
folds.
Then there exists a constant \(C>0\) such that, with probability at
least \(1-\gamma\) for all \(\gamma\in(0,1)\),
\begin{align} \label{eq:sroscar-bound}
\Delta_2^2(\hat\tau_{\text{SR}},\tau^r)
&\le
\Delta_2^2(\cD,\tau^r)
+
C\Biggl[
    B_{\mathrm{sh}}^2
    + \Rad_{n^r}^2(\cD)
\\ \nonumber
&\qquad
    + \sum_{a\in\{-1,1\}}
      \Bigl(
        \Rad_{n^o}^2(\cM^{o,\mathrm{sh}}_{a})
        +
        \Rad_{n^r}^2(\cD^{\mathrm{sh}}_{a})
      \Bigr)
\Biggr]
+ C\,\frac{\log(1/\gamma)}{n^r}.
\end{align}
\end{theorem}
The first term in \eqref{eq:sroscar-bound} is the approximation error
of the CATE calibration class \(\cD\), which is shared across RACER,
SR-OSCAR, and MR-OSCAR.
The second group of terms is purely statistical: it depends on the
complexities of the RCT calibration class \(\cD\), the OS outcome
classes \(\cM^{o,\mathrm{sh}}_{a}\), and the arm-wise RCT discrepancy classes
\(\cD^{\mathrm{sh}}_{a}\).
The mismatch penalty \(B_{\mathrm{sh}}^2\) in \eqref{eq:BZ-def} captures the
irreducible error induced by restricting to the shared space \(\bZ\):
even with infinite data and perfect optimization, SR-OSCAR cannot
improve beyond this term if important effect modifiers reside in
\(\bV\) that cannot be projected onto \(\bZ\).
Finally, the \(\log(1/\gamma)/n^r\) term is a standard residual due to
concentration.

\paragraph{Augmentation-error decomposition and localized rates.}

The proof of Theorem~\ref{thm:sroscar} proceeds in two conceptually
separate steps, mirroring the analysis in
\citet{asiaee2023leveraging}.
First, one shows that the CATE risk of SR-OSCAR can be bounded in
terms of the augmentation errors:
\begin{equation}
  \Delta_2^2(\hat\tau_{\text{SR}},\tau^r)
  \;\lesssim\;
  \Delta_2^2(\cD,\tau^r)
  +
  \Biggl(
    1 + \sum_{a\in\{-1,1\}}
            \Delta_2^2(m_{\bZ,a},\tilde\mu^r_a\mid\bX^r)
  \Biggr)\Rad_{n^r}(\cD),
  \label{eq:sroscar-aug-layer}
\end{equation}
up to negligible residual terms.
Second, one bounds the augmentation errors themselves in terms of the
OS/RCT nuisance complexities and the mismatch penalty \(B_{\mathrm{sh}}^2\).

\begin{proposition}[Augmentation error for shared-only borrowing]
\label{prop:sroscar-aug}
Under the conditions of Theorem~\ref{thm:sroscar}, there exist
constants \(C<\infty\), \(\eta_o>0\), and \(\eta_r>0\) such that, for
each arm \(a\),
\begin{equation}
  \Delta_2^2(m_{\bZ,a},\tilde\mu^r_a\mid\bX^r)
  \;\lesssim\;
  \frac{c(\cM^{o,\mathrm{sh}}_{a})}{(n^o)^{\eta_o}}
  +
  \frac{c(\cD^{\mathrm{sh}}_{a})}{(n^r)^{\eta_r}}
  +
  B_{\mathrm{sh}}^2,
  \label{eq:sroscar-aug-bound}
\end{equation}
where \(c(\cM^{o,\mathrm{sh}}_{a})\) and \(c(\cD^{\mathrm{sh}}_{a})\) are localized
complexities of the OS outcome and RCT discrepancy classes, respectively.
\end{proposition}
Combining \eqref{eq:sroscar-aug-layer} with
\eqref{eq:sroscar-aug-bound} and the generic localized-complexity bound
\eqref{eq:localized-generic} yields a localized version of
\eqref{eq:sroscar-bound}:
\begin{equation}
  \Delta_2^2(\hat\tau_{\text{SR}},\tau^r)
  \;\lesssim\;
  \Delta_2^2(\cD,\tau^r)
  +
  B_{\mathrm{sh}}^2
  +
  \frac{c(\cD)}{(n^r)^{\eta_r}}
  +
  \sum_{a\in\{-1,1\}}
  \biggl\{
    \frac{c(\cM^{o,\mathrm{sh}}_{a})}{(n^o)^{\eta_o}}
    +
    \frac{c(\cD^{\mathrm{sh}}_{a})}{(n^r)^{\eta_r}}
  \biggr\}.
  \label{eq:sroscar-localized}
\end{equation}
This form closely parallels the RACER bound
\eqref{eq:racer-localized} but makes explicit the additional mismatch
penalty \(B_{\mathrm{sh}}^2\) and the role of the OS outcome and calibration
complexities. In particular, when the outcome depends weakly on OS-only covariates
(so that \(B_{\mathrm{sh}}^2\) is small) and the OS outcome models are
well-behaved (small \(c(\cM^{o,\mathrm{sh}}_{a})\)), SR-OSCAR can strictly
reduce the CATE risk relative to RACER by reducing the estimation error
without incurring appreciable mismatch penalty.
Conversely, when \(B_{\mathrm{sh}}^2\) is large, the bound makes explicit that
restriction to the shared space can dominate the potential gain from
borrowing.

\subsection*{Proofs of Theorems}

\begin{proof}[Proof of~\Cref{prop:mroscar-aug}]
Fix an arm \(a\). Let \(g^\star\in\cG\) be an oracle imputation map achieving
\(r_{\mathrm{im}}^2=\E[\|\bV-g^\star(\bZ)\|^2\mid S=r]\), and write \(\widehat\bV:=g^\star(\bZ)\).
The MR-OSCAR arm-specific calibrated predictor evaluated at the imputed covariates is
\[
  m_{\mathrm{aug},a}(\bX^r)
    := \hat\mu^{o,\mathrm{im}}_{a}(\bZ,\widehat\bV)
       + \hat\delta^{\mathrm{im}}_{a}(\bX^r,\widehat\bV).
\]
Let \(\mu^{o,\mathrm{im},\star}_{a}\in\cM^{o,\mathrm{im}}_{a}\) and
\(\delta^{\mathrm{im},\star}_{a}\in\cD^{\mathrm{im}}_{a}\) denote oracle targets for the two ERM steps (so that, under correct specification,
\(\mu^r_a(\bm x)=\mu^{o,\mathrm{im},\star}_{a}(\bm x^o)+\delta^{\mathrm{im},\star}_{a}(\bm x)\)).
Add and subtract the oracle predictor at \(\widehat\bV\):
\begin{align*}
  m_{\mathrm{aug},a}(\bX^r)-\tilde\mu^r_a(\bX^r)
  &=
  \{\hat\mu^{o,\mathrm{im}}_{a}-\mu^{o,\mathrm{im},\star}_{a}\}(\bZ,\widehat\bV)
  + \{\hat\delta^{\mathrm{im}}_{a}-\delta^{\mathrm{im},\star}_{a}\}(\bX^r,\widehat\bV) \\
  &\quad
  + \{\mu^r_a(\bX^r,\widehat\bV)-\tilde\mu^r_a(\bX^r)\}.
\end{align*}
By \((x+y+z)^2\le 3(x^2+y^2+z^2)\) and Assumption~\ref{as:rates}, the first two terms contribute at most
\(\lesssim c(\cM^{o,\mathrm{im}}_{a})/(n^o)^{\eta_o}\) and \(\lesssim c(\cD^{\mathrm{im}}_{a})/(n^r)^{\eta_r}\) to
\(\Delta_2^2(m_{\mathrm{aug},a},\tilde\mu^r_a\mid \bX^r)\).
For the last term, use Jensen’s inequality and Lipschitzness of \(\mu^r_a\) in the \(\bV\) block (Assumption~\ref{as:lipschitz}):
\[
  \bigl|\mu^r_a(\bX^r,\widehat\bV)-\tilde\mu^r_a(\bX^r)\bigr|
  =
  \Bigl|\E\bigl\{\mu^r_a(\bX^r,\widehat\bV)-\mu^r_a(\bX^r,\bV)\mid \bX^r,S=r\bigr\}\Bigr|
  \le L\,\E[\|\widehat\bV-\bV\|\mid \bX^r,S=r],
\]
which implies \(\E[\{\mu^r_a(\bX^r,\widehat\bV)-\tilde\mu^r_a(\bX^r)\}^2\mid S=r]\le L^2\,\E[\|\widehat\bV-\bV\|^2\mid S=r]=L^2 r_{\mathrm{im}}^2\).
Combining bounds yields \eqref{eq:mroscar-aug-bound}.
\end{proof}

\begin{proof}[Proof of~\Cref{prop:sroscar-aug}]
Fix an arm \(a\). In SR-OSCAR the arm-specific calibrated predictor is
\(
  m_{\bZ,a}(\bX^r)=\hat\mu^{o,\mathrm{sh}}_{a}(\bZ)+\hat\delta^{\mathrm{sh}}_{a}(\bX^r).
\)
Let \(h_a^\star\in\cH^{\mathrm{sh}}_{a}\) be an oracle element attaining (or nearly attaining) the infimum in the definition of \(B_{\mathrm{sh}}^2\), and write
\(h_a^\star(\bX^r)=\mu^{o,\mathrm{sh},\star}_{a}(\bZ)+\delta^{\mathrm{sh},\star}_{a}(\bX^r)\) with
\(\mu^{o,\mathrm{sh},\star}_{a}\in\cM^{o,\mathrm{sh}}_{a}\) and \(\delta^{\mathrm{sh},\star}_{a}\in\cD^{\mathrm{sh}}_{a}\).
Then
\begin{align*}
  \Delta_2^2(m_{\bZ,a},\tilde\mu^r_a\mid \bX^r)
  &\le
  2\,\E\bigl[(m_{\bZ,a}(\bX^r)-h_a^\star(\bX^r))^2\mid S=r\bigr]
  +2\,\E\bigl[(h_a^\star(\bX^r)-\tilde\mu^r_a(\bX^r))^2\mid S=r\bigr].
\end{align*}
The second term is bounded by the (single-arm) approximation residual defining \(B_{\mathrm{sh}}^2\), and hence by \(B_{\mathrm{sh}}^2\) itself.
For the first term, \((x+y)^2\le 2x^2+2y^2\) gives
\[
  \E[(m_{\bZ,a}-h_a^\star)^2\mid S=r]
  \le
  2\,\E[(\hat\mu^{o,\mathrm{sh}}_{a}(\bZ)-\mu^{o,\mathrm{sh},\star}_{a}(\bZ))^2\mid S=r]
  +2\,\E[(\hat\delta^{\mathrm{sh}}_{a}(\bX^r)-\delta^{\mathrm{sh},\star}_{a}(\bX^r))^2\mid S=r].
\]
By Assumption~\ref{as:rates} (equivalently \eqref{eq:localized-generic}), these two estimation errors are
\(\lesssim c(\cM^{o,\mathrm{sh}}_{a})/(n^o)^{\eta_o}\) and \(\lesssim c(\cD^{\mathrm{sh}}_{a})/(n^r)^{\eta_r}\).
Combining the pieces yields \eqref{eq:sroscar-aug-bound}.
\end{proof}

\begin{proof}[Proof of Theorem~\ref{thm:mroscar}]
We define \(\hat\phi_i\) as the doubly-robust pseudo-outcome for each RCT unit \(i\) as described in \eqref{eq:pseudo-outcome}. Recall that the MR-OSCAR estimator $\hat\tau_{\text{MR}}$ is obtained as
the empirical risk minimizer over $\cD$ of a squared-loss regression of
cross-fitted pseudo-outcomes $\hat\phi_i$ on \(\bX_i^r\) in the RCT.
Throughout the proof we condition on the splits used for cross-fitting.

\paragraph{Step 1: Oracle pseudo-outcomes and generic CATE bound.}

Let $\phi_i^\star$ denote the \emph{oracle} pseudo-outcome that would be
constructed if all nuisance components (outcome regressions,
discrepancies, imputation map) were known.  By
Assumption~\ref{as:ident}, the oracle CATE $\tau^r$ minimizes the
population squared error
$L(f) := \E\bigl\{ (\phi^\star - f(\bX^r))^2 \mid S = r \bigr\}$ over
$f\in L_2(P^r)$, and the excess CATE risk
$\Delta_2^2(\hat\tau_{\text{MR}},\tau^r)$ coincides (up to constants) with the
excess risk of regressing $\phi^\star$ on \(\bX^r\); see
\citet{asiaee2023leveraging} for a detailed derivation.

Let $\hat\tau^{\mathrm{or}}$ be the empirical risk minimizer in $\cD$
when regressing the oracle pseudo-outcomes $(\phi_i^\star)$ on the
cross-fitted RCT sample.  Standard ERM arguments with squared loss and
Rademacher complexity (\citealp{bartlett2006local}) then yield, for some constant $C>0$ and
all $\gamma\in(0,1)$,
\begin{equation}
  \Delta_2^2(\hat\tau^{\mathrm{or}},\tau^r)
  \;\le\;
  \Delta_2^2(\cD,\tau^r)
  +
  C\Bigl[
      \Rad_{n^r}^2(\cD)
      + \frac{\log(1/\gamma)}{n^r}
    \Bigr],
  \label{eq:mr-oracle-bound}
\end{equation}
with probability at least $1-\gamma/2$.
Here $\Delta_2^2(\cD,\tau^r)$ is the approximation error of $\cD$ and
$\Rad_{n^r}(\cD)$ is the empirical Rademacher complexity of $\cD$ under
the RCT distribution.

\paragraph{Step 2: Decomposing the effect of estimated pseudo-outcomes.}

The actual estimator $\hat\tau_{\text{MR}}$ minimizes the empirical risk
based on \emph{estimated} pseudo-outcomes $\hat\phi_i$, obtained by
plugging cross-fitted nuisance estimators into the oracle construction.
Let
\[
  \ell(f,\phi)
    := (\phi - f(\bX^r))^2
\]
denote the squared loss.  By a standard perturbation argument for
plug-in pseudo-outcome regressions, we have
\begin{align}
  \Delta_2^2(\hat\tau_{\text{MR}},\tau^r)
  &\le
  2\,\Delta_2^2(\hat\tau^{\mathrm{or}},\tau^r)
  +
  C\,\E\Bigl[
    \bigl(\hat\phi - \phi^\star\bigr)^2 \mid S = r
  \Bigr],
  \label{eq:mr-split-oracle}
\end{align}
for some universal constant $C>0$.
Intuitively, the first term is the excess risk due to restricting to the
class $\cD$ (already controlled in \eqref{eq:mr-oracle-bound}),
whereas the second term quantifies the extra noise introduced by using
estimated, rather than oracle, pseudo-outcomes.

The error between $\hat\phi$ and $\phi^\star$ can be expressed in terms
of the arm-specific augmentation functions that MR-OSCAR estimates.
Writing $m_{\mathrm{aug},a}(\bX^r)$ for the implicit augmentation
function for arm $a$ (as defined in the statement of
Proposition~\ref{prop:mroscar-aug}) and using the algebra of the
pseudo-outcome representation (exactly as in Equation~(7.23) of
\citealp{asiaee2023leveraging}), one obtains
\begin{equation}
  \E\Bigl[
    \bigl(\hat\phi - \phi^\star\bigr)^2 \mid S = r
  \Bigr]
  \;\lesssim\;
  \sum_{a\in\{-1,1\}}
    \Delta_2^2\bigl(
      m_{\mathrm{aug},a},
      \tilde\mu^r_a \mid \bX^r
    \bigr),
  \label{eq:mr-phi-vs-oracle}
\end{equation}
where $\tilde\mu^r_a$ is the arm-specific target appearing in the
augmentation layer.
Combining \eqref{eq:mr-split-oracle},
\eqref{eq:mr-oracle-bound}, and \eqref{eq:mr-phi-vs-oracle}, and
absorbing constants, yields
\begin{align}
  \Delta_2^2(\hat\tau_{\text{MR}},\tau^r)
  &\le
  C\Biggl[
    \Delta_2^2(\cD,\tau^r)
    + \Rad_{n^r}^2(\cD)
    + \frac{\log(1/\gamma)}{n^r}
    + \sum_{a\in\{-1,1\}}
        \Delta_2^2\bigl(
          m_{\mathrm{aug},a},
          \tilde\mu^r_a \mid \bX^r
        \bigr)
  \Biggr],
  \label{eq:mr-pre-prop}
\end{align}
with probability at least $1-\gamma$ (after a union bound over the
oracle and plug-in steps).

\paragraph{Step 3: Bounding the augmentation errors via
Proposition~\ref{prop:mroscar-aug}.}

Under the conditions of Theorem~\ref{thm:mroscar},
Proposition~\ref{prop:mroscar-aug} guarantees that, for each arm $a$,
\[
  \Delta_2^2\bigl(
    m_{\mathrm{aug},a},
    \tilde\mu^r_a \mid \bX^r
  \bigr)
  \;\lesssim\;
  \frac{c(\cM^{o,\mathrm{im}}_{a})}{(n^o)^{\eta_o}}
  +
  \frac{c(\cD^{\mathrm{im}}_{a})}{(n^r)^{\eta_r}}
  +
  L^2 r_{\mathrm{im}}^2,
\]
where $c(\cdot)$ and $\eta_o,\eta_r>0$ are localized complexity
constants associated with the corresponding function classes, and
$r_{\mathrm{im}}^2$ is the oracle imputation risk defined in
Section~\ref{subsec:theory-mroscar}.
Assumption~\ref{as:rates} states that the localized complexities can be
dominated (up to constants) by the squared Rademacher complexities of
the same classes, so that
\[
  \frac{c(\cM^{o,\mathrm{im}}_{a})}{(n^o)^{\eta_o}}
  \;\lesssim\;
  \Rad_{n^o}^2(\cM^{o,\mathrm{im}}_{a}),
  \qquad
  \frac{c(\cD^{\mathrm{im}}_{a})}{(n^r)^{\eta_r}}
  \;\lesssim\;
  \Rad_{n^r}^2(\cD^{\mathrm{im}}_{a}).
\]
Summing the bound in Proposition~\ref{prop:mroscar-aug} over
$a\in\{-1,1\}$ and substituting into \eqref{eq:mr-pre-prop} yields
\begin{align}
  \Delta_2^2(\hat\tau_{\text{MR}},\tau^r)
  \;\le\;
  C\Biggl[
    \Delta_2^2(\cD,\tau^r)
    + \Rad_{n^r}^2(\cD)
    + \frac{\log(1/\gamma)}{n^r}
    + L^2 r_{\mathrm{im}}^2
    + \sum_{a\in\{-1,1\}}
        \Bigl\{
          \Rad_{n^o}^2(\cM^{o,\mathrm{im}}_{a})
          +
          \Rad_{n^r}^2(\cD^{\mathrm{im}}_{a})
        \Bigr\}
  \Biggr],
  \label{eq:mr-after-prop}
\end{align}
again up to universal constants.

\paragraph{Step 4: Error from estimating the imputation map.}

So far $r_{\mathrm{im}}^2$ has been defined relative to the oracle imputation
map $g^\star\in\cG$.
In practice, MR-OSCAR uses $\hat g$, obtained by ERM (or penalized ERM)
in the same class $\cG$ on the OS sample with cross-fitting.
Assumption~\ref{as:rates} implies
\[
  \E\bigl[\|\hat g(\bZ) - g^\star(\bZ)\|^2 \mid S = o\bigr]
  \;\lesssim\;
  \Rad_{n^o}^2(\cG).
\]
By the Lipschitz condition in Assumption~\ref{as:lipschitz}, replacing
$g^\star$ by $\hat g$ in the construction of the pseudo-outcomes and
augmentation functions perturbs the latter by at most a constant
multiple of $\|\hat g(\bZ)-g^\star(\bZ)\|$.
Therefore the imputation-induced part of the augmentation error picks up
an additional term of order $\Rad_{n^o}^2(\cG)$, which can be absorbed
additively into the right-hand side of \eqref{eq:mr-after-prop}.

Collecting terms and re-absorbing universal constants into a single $C$
gives exactly the bound in \eqref{eq:mroscar-bound}:
\[
  \Delta_2^2(\hat\tau_{\text{MR}},\tau^r)
  \;\le\;
  \Delta_2^2(\cD,\tau^r)
  +
  C\Biggl[
      L^2 r_{\mathrm{im}}^2
      + \Rad_{n^r}^2(\cD)
      + \sum_{a\in\{-1,1\}}
          \Bigl(
            \Rad_{n^o}^2(\cM^{o,\mathrm{im}}_{a})
            +
            \Rad_{n^r}^2(\cD^{\mathrm{im}}_{a})
          \Bigr)
      + \Rad_{n^o}^2(\cG)
    \Biggr]
  + C\,\frac{\log(1/\gamma)}{n^r},
\]
with probability at least $1-\gamma$.
This completes the proof.
\end{proof}

\begin{proof}[Proof of~\Cref{thm:MR-linear}]
We have covariates $\bZ\in\mathbb R^{p_z}$, target to impute $\bV\in\mathbb R^{p_v}$. Linear-Gaussian relations in each source $s\in\{o,r\}$:
$$
  \bV^s \;=\; \bm\Lambda_s \bZ^s + \varepsilon^s,\qquad
  \mathbb E[\varepsilon^s\mid \bZ^s]=0,\quad
  \mathrm{Cov}(\varepsilon^s)=\bm\Sigma_{\bV\mid \bZ}^{\,s},\quad
  \mathrm{Cov}(\bZ^s)=\Sigma_{\bZ\bZ}^{\,s}.
  $$
We estimate $\bm\Lambda_o$ row-wise by LASSO on OS data and then impute in RCT as
  $\widehat \bV(\bZ^r) = \widehat{\bm\Lambda}_o\,\bZ^r$. The RCT mean-squared imputation error is
  $
  r_{\mathbf \bV}^2 \;:=\; \mathbb E_r\!\big[\;\|\widehat \bV(\bZ^r) - \bV^r\|_2^2\;\big].
  $ Let $\bm\Lambda_o^\star$ be the best row-sparse approximation to $\bm\Lambda_o$ (defined precisely below), and define the covariance transfer factor
  $$
  \kappa_{r\!\leftarrow\!o} \;:=\; \big\|\,\bm\Sigma_{\bZ\bZ}^{\,r}(\Sigma_{\bZ\bZ}^{\,o})^{-1}\,\big\|_{\mathrm{op}}
  \;=\; \big\|\,(\bm\Sigma_{\bZ\bZ}^{\,o})^{-1/2}\,\bm\Sigma_{\bZ\bZ}^{\,r}\,(\bm\Sigma_{\bZ\bZ}^{\,o})^{-1/2}\,\big\|_{\mathrm{op}}.
  $$
  Write $\bV^r=\bm\Lambda_r \bZ^r+\varepsilon^r$ and $\widehat {\bV}(\bZ^r)=\widehat{\bm\Lambda}_o \bZ^r$. Then,
$
\widehat {\bV}(\bZ^r)-\bV^r
\;=\;
(\widehat{\bm\Lambda}_o-\bm{\Lambda}_r)\bZ^r - \varepsilon^r.
$ Hence
$$
\begin{aligned}
r_{\mathrm{im}}^2
&= \mathbb E\!\left[\;\|(\widehat{\bm\Lambda}_o-\bm\Lambda_r)\bZ^r - \varepsilon^r\|_2^2 \mid S = r\;\right] \\
&= \mathbb E\!\left[\;\|(\widehat{\bm\Lambda}_o-\bm\Lambda_r)Z^r\|_2^2 \mid S = r\;\right]
	   \;-\;2\,\mathbb E\!\left[\; \langle (\widehat{\bm\Lambda}_o-\bm\Lambda_r)\bZ^r,\ \varepsilon^r\rangle \mid S = r\;\right]
	   \;+\;\mathbb E\!\left[\;\|\varepsilon^r\|_2^2 \mid S = r\;\right].
\end{aligned}
$$
Since $\mathbb E[\varepsilon^r\mid \bZ^r, S = r]=0$ and $\varepsilon^r\perp \bZ^r$, the cross term is zero:
$$
\mathbb E\!\left[\; \langle (\widehat{\bm\Lambda}_o-\bm\Lambda_r)\bZ^r,\ \varepsilon^r\rangle \mid S = r\;\right]
= \mathbb E_r\!\left[\, \mathbb E_r\big[\langle (\widehat{\bm\Lambda}_o-\bm\Lambda_r)\bZ^r,\ \varepsilon^r\rangle\mid \bZ^r\big] \,\right]=0.
$$
Also $\mathbb E[\|\varepsilon^r\|_2^2 \mid S = r]=\mathrm{tr}(\bm\Sigma_{\bV\mid \bZ}^{\,r})$. Therefore
$$
{\;
r_{\mathrm{im}}^2
\;=\;
\mathrm{tr}\!\big((\widehat{\bm\Lambda}_o-\bm\Lambda_r)\,\bm\Sigma_{\bZ\bZ}^{\,r}\,(\widehat{\bm\Lambda}_o-\bm\Lambda_r)^\mathrm{T}\big)
\;+\;
\mathrm{tr}(\bm\Sigma_{\bV\mid \bZ}^{\,r}).
\;}
$$
Next, we fix any row-sparse $\bm\Lambda_o^\star$ (we will choose it as the best $s$-sparse approximation to $\bm\Lambda_o$ in OS prediction norm). Define
$
\bm\Delta \;:=\; \widehat{\bm\Lambda}_o - \bm\Lambda_o^\star,\;
\bm B \;:=\; \bm\Lambda_o^\star - \bm\Lambda_r.
$ Then $\widehat{\bm\Lambda}_o - \bm\Lambda_r = \bm\Delta + \bm B$. Plugging it into the quadratic form, we get
$$
\begin{aligned}
\mathrm{tr}\!\Big((\widehat{\bm\Lambda}_o-\bm\Lambda_r)\,\bm\Sigma_{\bZ\bZ}^{\,r}\,(\widehat{\bm\Lambda}_o-\bm\Lambda_r)^\mathrm{T}\Big)
&= \mathrm{tr}\!\Big((\bm\Delta+\bm B)\,\bm \Sigma_{\bZ\bZ}^{\,r}\,(\bm\Delta+\bm B)^\mathrm{T}\Big) \\
&= \mathrm{tr}\!\big(\bm\Delta \bm\Sigma_r \bm\Delta^\mathrm{T}\big)
 + \mathrm{tr}\!\big(\bm B \bm\Sigma_r \bm B^\mathrm{T}\big)
 + 2\,\mathrm{tr}\!\big(\bm\Delta\bm\Sigma_r \bm B^\mathrm{T}\big),
\end{aligned}
$$
where $\bm\Sigma_r:=\bm\Sigma_{\bZ\bZ}^{\,r}$. Now applying Cauchy-Schwarz for the Frobenius inner product with the PSD weight $\bm\Sigma_r$, we get for any conformable matrices $A_1,A_2$,
$$
\mathrm{tr}(A_1\bm\Sigma_r A_2^\mathrm{T})
= \langle A_1\bm\Sigma_r^{1/2},\,A_2\bm\Sigma_r^{1/2}\rangle_F
\;\le\; \|A_1\bm\Sigma_r^{1/2}\|_F\,\|A_2\bm\Sigma_r^{1/2}\|_F
= \big\{\mathrm{tr}(A_1\bm\Sigma_r A_1^\mathrm{T})\,\mathrm{tr}(A_2\bm\Sigma_r A_2^\mathrm{T})\big\}^{1/2}.
$$
Therefore
$
2\,\mathrm{tr}(\bm\Delta\bm\Sigma_r \bm B^\mathrm{T})
\;\le\; 2\,\sqrt{\mathrm{tr}(\bm\Delta \bm\Sigma_r\Delta^\mathrm{T})\,\mathrm{tr}(\bm B \bm\Sigma_r \bm B^\mathrm{T})}
\;\le\;
\mathrm{tr}(\bm\Delta \bm\Sigma_r \bm\Delta^\mathrm{T})+\mathrm{tr}(\bm B \bm\Sigma_r \bm B^\mathrm{T}),
$ using $2ab\le a^2+b^2$. Hence
$$
\mathrm{tr}\!\Big((\widehat{\bm\Lambda}_o-\bm\Lambda_r)\bm\Sigma_r(\widehat{\bm\Lambda}_o-\bm\Lambda_r)^\mathrm{T}\Big)
\;\le\;
2\,\mathrm{tr}(\bm B\bm\Sigma_r \bm B^\mathrm{T})
\;+\;
2\,\mathrm{tr}(\bm\Delta \bm\Sigma_r \bm\Delta^\mathrm{T}).
$$
Combining, we have
$$
{\;
r_{\mathrm{im}}^2
\;\le\;
2\,\underbrace{\mathrm{tr}\!\big(\bm B \bm\Sigma_r \bm B^\mathrm{T}\big)}_{T_{\text{shift}}}
\;+\;
2\,\underbrace{\mathrm{tr}\!\big(\bm\Delta\bm\Sigma_r\bm\Delta^\mathrm{T}\big)}_{T_{\text{est}}}
\;+\;
\mathrm{tr}(\bm\Sigma_{\bV\mid \bZ}^{\,r}),
\;}
$$ where $\;T_{\text{shift}}=\mathrm{tr}\!\big((\bm\Lambda_o^\star-\bm\Lambda_r)\bm\Sigma_{\bZ\bZ}^{\,r}(\bm\Lambda_o^\star-\bm\Lambda_r)^\mathrm{T}\big)$,
and $T_{\text{est}}$ is the OS estimation contribution measured under the RCT covariance. Now, for any vector $x\neq 0$,
$$
\frac{x^\mathrm{T} \bm\Sigma_{\bZ\bZ}^{\,r} x}{x^\mathrm{T} \bm\Sigma_{\bZ\bZ}^{\,o} x}
\;=\;
\frac{y^\mathrm{T}\big((\bm\Sigma_{\bZ\bZ}^{\,o})^{-1/2}\bm\Sigma_{\bZ\bZ}^{\,r}(\bm\Sigma_{\bZ\bZ}^{\,o})^{-1/2}\big)y}{y^\mathrm{T} y}
\;\le\;
\big\|(\bm\Sigma_{\bZ\bZ}^{\,o})^{-1/2}\bm\Sigma_{ZZ}^{\,r}(\bm\Sigma_{\bZ\bZ}^{\,o})^{-1/2}\big\|_{\mathrm{op}}
\;=\;\kappa_{r\!\leftarrow\!o},
$$
with $y=(\bm\Sigma_{\bZ\bZ}^{\,o})^{1/2}x$. Thus $\bm\Sigma_{\bZ\bZ}^{\,r}\preceq \kappa_{r\!\leftarrow\!o}\,\bm\Sigma_{\bZ\bZ}^{\,o}$ (PSD order), which implies,
$$
\mathrm{tr}(\bm\Delta\bm\Sigma_{\bZ\bZ}^{\,r}\bm\Delta^\mathrm{T}) \;\le\; \kappa_{r\!\leftarrow\!o}\,\mathrm{tr}(\bm\Delta\Sigma_{ZZ}^{\,o}\bm\Delta^\mathrm{T}).
$$
Suppose, we estimate each row of $\Lambda_o$ by LASSO on OS. For the $j$-th response coordinate $V_j$,
$$
v^o_{\,j} = \bZ^o\lambda_{o,j} + \varepsilon^o_j,\qquad
\widehat\lambda_{o,j}\in\arg\min_{\beta\in\mathbb R^{p_z}}
\Big\{ \tfrac{1}{2\,n^o}\|v^o_{\,j}-\bZ^o\beta\|_2^2 + \lambda_j\|\beta\|_1\Big\}.
$$
Let $\lambda_{o,j}^\star$ be the best $s_j$-sparse approximation to $\lambda_{o,j}$ in the OS prediction norm:
$$
\lambda_{o,j}^\star\in\arg\min_{\|\beta\|_0\le s_j} (\beta-\lambda_{o,j})^\mathrm{T} \Sigma_{ZZ}^{\,o}(\beta-\lambda_{o,j}).
$$
Set $\Delta_j := \widehat\lambda_{o,j}-\lambda_{o,j}^\star$. Under the usual restricted eigenvalue (RE) condition holds on the support $S_j=\mathrm{supp}(\lambda_{o,j}^\star)$ with constant $\phi_0>0$, and for
$\lambda_j \asymp \sigma_{o,j}\sqrt{(2\log p_z)/n^o}$, where $\sigma_{o,j}^2:=(\Sigma_{\bV\mid \bZ}^{\,o})_{jj}$,
the standard LASSO theory (KKT + RE + Gaussian/sub-Gaussian tails) yields, with probability at least $1-c_1p_z^{-c_2}$,
$$
\Delta_j^\mathrm{T} \bm\Sigma_{\bZ\bZ}^{\,o}\,\Delta_j
 \;\le\;
C_0\,\frac{\lambda_j^2\,s_j}{\phi_0^2}
\;\le\;
C_1\,\frac{s_j\log p_z}{n^o}\,\sigma_{o,j}^2,
$$
where we used $\lambda_j^2 \asymp \sigma_{o,j}^2 (\log p_z)/n^o$ in the last step. Now stacking all rows such that $\Delta$ has rows $\Delta_j^\mathrm{T}$, we have
$$
\mathrm{tr}\!\big(\bm\Delta\bm\Sigma_{\bZ\bZ}^{\,o}\bm\Delta^\mathrm{T}\big)
= \sum_{j=1}^{p_v} \Delta_j^\mathrm{T} \bm\Sigma_{\bZ\bZ}^{\,o}\Delta_j
\;\le\;
C_1\,\frac{\log p_z}{n^o}\,\sum_{j=1}^{p_v} s_j\,\sigma_{o,j}^2.
$$
Renaming $C:=2C_1$ and combining all pieces, we finally have $$r_{\mathrm{im}}^2
\ \le\
2\,\underbrace{\mathrm{tr}\!\bigl((\bm\Lambda_o^\star-\bm\Lambda_r)\,\bm\Sigma_{\bZ\bZ}^{\,r}\,(\bm\Lambda_o^\star-\bm\Lambda_r)^\top\bigr)}_{\text{OS}\to\text{RCT mean-relation shift + sparsity approximation}}
\;+\;
\underbrace{C\,\kappa_{r\!\leftarrow\!o}\,\frac{\log p_z}{n^o}\,\sum_{j=1}^{p_v}s_j\,\sigma_{o,j}^2}_{\text{OS estimation error (sparse LASSO)}}
\;+\;
\underbrace{\mathrm{tr}(\bm\Sigma_{\bV\mid \bZ}^{\,r})}_{\text{irreducible RCT noise}}$$
\end{proof}

\end{document}